\documentclass[review]{elsarticle}

\usepackage{lineno}
\usepackage{hyperref}
\modulolinenumbers[5]

\journal{Computer Networks}

\newcommand{\nop}[1]{} 
\usepackage{amsmath,amssymb,xfrac}

\usepackage{cuted}
\usepackage[Algorithm,ruled]{algorithm}
\usepackage{algpseudocode}
\usepackage{subfigure}
\usepackage{siunitx}
\usepackage{mathtools}
\usepackage{morefloats}

\usepackage{etoolbox}

\newtheorem{definition}{Definition}

\newtheorem{lemma}{Lemma}

\newtheorem{example}{Example}

\newcommand{\signed}%
{{\unskip\nobreak\hfill\penalty50
    \hskip2em\hbox{}\nobreak\hfil $\blacksquare$
    \parfillskip=0pt \finalhyphendemerits=0 \par}}
\newenvironment{proof}[1]
{
  \noindent\bf{Proof: }\rm{#1}\ignorespaces
}
{\signed\addvspace\medskipamount}

\graphicspath{{./figures/}}

\begin{document}
	
	\pagestyle{empty}

\begin{frontmatter}

\title{Local Voting: A New Distributed Bandwidth Reservation Algorithm for
       6TiSCH Networks}






\author[kastoria]{Dimitrios~J.~Vergados\corref{mycorrespondingauthor}}
\cortext[mycorrespondingauthor]{Corresponding author}
\ead{dvergados@uowm.gr}
\author[ntnu]{Katina Kralevska}
\ead{katinak@ntnu.no}
\author[ntnu]{Yuming~Jiang}
\ead{jiang@ntnu.no}
\author[kastoria]{Angelos~Michalas}
\ead{amichalas@uowm.gr}

\address[kastoria]{Dep.\ of~Informatics, University~of~Western~Macedonia,
                   Kastoria, Greece}
\address[ntnu]{Dep.\ of Information Security and Communication Technology, NTNU,
               Norwegian~University~of~Science~and~Technology}

\begin{abstract}
The IETF 6TiSCH working group fosters the adaptation of IPv6-based protocols
into Internet of Things by introducing the 6TiSCH Operation Sublayer (6top).
The 6TiSCH architecture integrates the high reliability and low-energy
consumption of IEEE~802.15.4e Time Slotted Channel Hopping (TSCH) with IPv6.
IEEE~802.15.4e TSCH defines only the communication between nodes through a
schedule but it does not specify how the resources are allocated for
communication between the nodes in 6TiSCH networks.
We propose a distributed algorithm for bandwidth allocation, called Local
Voting, that adapts the schedule to the network conditions.
The algorithm tries to equalize the link load (defined as the ratio of the
queue length plus the new packet arrivals, over the number of allocated cells)
through cell reallocation by
calculating the number of cells to be added or released by 6top.
Simulation results show that equalizing the load throughout 6TiSCH network
provides better fairness in terms of load, reduces the queue sizes and packets
reach the root faster compared to representative algorithms from the literature.
Local Voting combines good delay performance and energy efficiency that are
crucial features for Industrial Internet-of-Things applications.
\end{abstract}


\end{frontmatter}

\nolinenumbers

\section{Introduction}
Wireless Sensor Networks (WSNs) have advanced significantly in the past decades.
The recent increase of connected devices has triggered countless
Internet-of-Things (IoT) applications to emerge~\cite{ATZORI20102787}.
It is expected that 50 billion devices will be connected to the Internet by
2020~\cite{cisco}.
The so-called Industrial Internet-of-Things (IIoT) is modernizing various
domains such as home automation, transportation, manufacturing, agriculture, and
other industrial sectors.

Often IoT is realized through Low-power and Lossy Networks (LLNs), which
consist of low complexity resource constrained embedded devices, that are
interconnected using different wireless technologies.
The IEEE 802.15.4e standard defines the physical and the medium access control
(MAC) layers for ultra-low power and reliable networking solutions for
LLNs~\cite{series/asc/GuglielmoAS14}.
There are five MAC modes: Time Slotted Channel Hopping (TSCH), Deterministic
and Synchronous Multi-channel Extension (DSME), Low Latency Deterministic
Network (LLDN), Asynchronous Multi-Channel Adaptation (AMCA), and Radio
Frequency Identification Blink (BLINK)~\cite{rfc7554}.
In this work, we study TSCH which is designed to allow IEEE 802.15.4 devices to
support a wide range of applications, including industrial ones.
In industrial environments, the large metallic equipment causes multi-path
fading and interference~\cite{rfc}, and TSCH combats against them by combining
channel hopping and time synchronization.
The channel hopping allows transmissions between nodes to use different
channels, while the slotted access enhances the reliability by synchronizing
the nodes with a schedule and, thus, avoiding collisions.

The IETF 6TiSCH working group standardizes the protocol stack for
IIoT~\cite{wang-6tisch-6top-sublayer-04}.
It combines the high reliability and the low-energy consumption of
IEEE 802.15.4e TSCH with the addressability and Internet integration
capabilities of Internet Protocol version 6 (IPv6).
The communication in a 6TiSCH network is orchestrated by a schedule composed of
cells, where each cell is identified by [slotOffset, channelOffset]~%
\cite{ietf-6tisch-6top-protocol-12}.
The schedule specifies the channel (based on the
channelOffset) and the time slot (based on the slotOffset) for communication of
a node with each of its neighbors.
The IEEE 802.15.4e standard defines how the schedule is executed but it does
not define how the schedule is built and updated.
Fig.~\ref{layer} shows the 6TiSCH protocol stack where the 6TiSCH Operation
Sublayer (6top) integrates the IEEE 802.15.4e MAC-TSCH layer with the
IPv6-enabled upper stack~\cite{ietf-6tisch-6top-protocol-12}.
The roles of the 6top sublayer are:
\begin{itemize}
   \item to terminate the 6top Protocol (6P), which allows a node to communicate
         with a neighboring node to add/delete cells;
   \item to run one or multiple 6top scheduling functions (SF), which define
         the rules when to add/delete cells between neighboring nodes while
         monitoring performance and collecting statistics.
\end{itemize}

\begin{figure}
  \centering
  \includegraphics[width=0.8\textwidth]{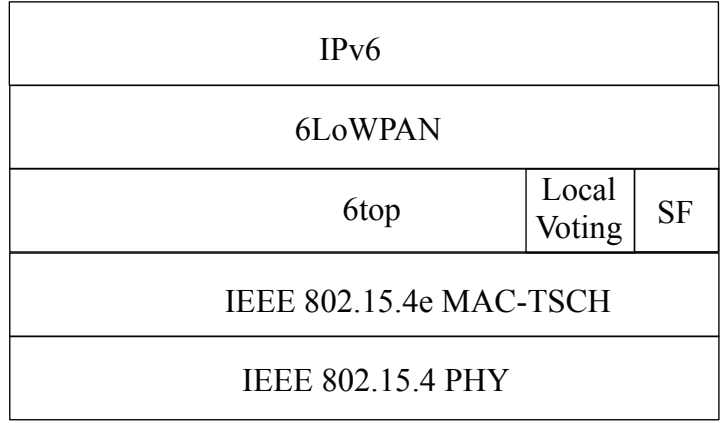}
  \caption{The 6TiSCH protocol stack. We propose an algorithm for bandwidth
           reservation, called Local Voting, that is located in the 6top
           sublayer.}
  \label{layer}
\end{figure}

The biggest challenges for enabling the pervasive deployment of IoT devices are
the demand for high reliability and the limited energy supply for the nodes.
These challenges are magnified with the increase of the number of network
devices and the emergence of new applications with diverse requirements.
As the deployment cases become more dense, and new applications and devices
are added, the traffic patterns become more congested.
In these conditions, we have found that the network performance is determined
by the ability of the network to distribute the resources (cells) among the
competing links, in a way that maximizes efficiency~\cite{8091101,8116537}.

In this paper, we propose a distributed bandwidth reservation algorithm
called~\emph{Local Voting} (LV).
It balances the load between the links in the network, where the load is
defined as the ratio of the queue length plus new packet arrvals,
over the number of allocated cells.
LV was originally proposed in~\cite{8091101} in the context of wireless mesh
networks.
Here we adapt the above algorithm for the link-based multi-channel
environment of 6TiSCH networks.
Through analysis and extensive performance evaluation we show here that by
redistributing cells
among the links, we can limit the maximum delay in the network, and at the
same time enhance reliability and fairness at a lower energy cost
compared to scenarios where no load balancing takes place.

Most of the related works focus on ways to construct an optimal schedule
between the links, without taking into consideration the optimal
number of cells that should be allocated to each link.
These works usually consider the On-The-Fly (OTF) bandwidth reservation
algorithm~\cite{7273816}.
Using the OTF algorithm, each node in the network estimates the number
of cells that it requires for fulfilling its communication requirements
by estimating the amount of new and forwarded traffic that it needs to
transmit to its parent nodes.
Then, the OTF module asks the 6top sublayer to add or remove cells, in order
for the allocation to match this number if possible.
However, the OTF algorithm does not consider the case where the requested
number of cells exceeds the number of available cells, due to congestion.
The nodes under OTF also do not consider the traffic requirements of the
neighboring nodes, so there is no provision for cell redistribution to
neighbors with higher bandwidth demands.
Finally, the OTF algorithm tries to maintain a stable schedule by using
a long-term average of the estimated throughput, which leads to inefficient
allocation when the traffic patterns fluctuate.
For these reasons we introduce LV algorithm that addresses the
above limitations of OTF\@.
We compare a thorough performance comparison between two versions of LV and OTF
and Enhanced-OTF (E-OTF)~\cite{8449793}.
Under LV, information about the queue lengths and the cell allocations are
periodically diffused among neighboring nodes, which use this information to
calculate the number of cells that should be allocated to each link based on
the load of each interfering link.
Equalizing the load throughout the congested areas in the network leads to
better fairness in terms of load for LV compared to OTF and E-OTF\@.
The performance evaluation also shows that LV provides similar performance in
terms of delay to E-OTF with an energy consumption similar to OTF, making LV a
promising distributed bandwidth reservation algorithm for 6TiSCH networks.

The rest of the paper is organized as follows. Section~\ref{rw} summarizes the
related works.
In Section~\ref{model}, the network model is formulated.
Section~\ref{lv} presents Local Voting algorithm.
Extensive performance evaluation results are presented in Section~\ref{perf},
and Section~\ref{conc} concludes the paper.

\section{Related Works} \label{rw}

\subsection{6TiSCH scheduling protocols}

Recently there have been many proposals for centralized and distributed
solutions for TSCH scheduling in the literature.
Centralized algorithms designate a specific scheduling entity that collects
information about the network and adjusts the TSCH schedule to it.
The first proposed centralized algorithm is Traffic Aware Scheduling Algorithm
(TASA)~\cite{6362805}, which builds a time/frequency collision-free schedule
in a centralized manner.
A master node collects information about the entire network topology and the
load of each node.
Then, it computes the schedule by exploiting matching and coloring procedures.
The main disadvantages of centralized scheduling techniques is the signaling
overhead since each node in the network has to communicate with the scheduler,
there is a single point of failure, and there is a limit on the
size of the topology since the scheduler becomes a bottleneck of the scheduling
function.

As a counterpart, distributed approaches have been proposed where nodes agree
on the schedule by applying distributed scheduling protocols and
neighbor-to-neighbor negotiation, without having a central entity.
The first distributed scheduling algorithm was proposed in~\cite{Tinka2010},
and it has been followed later by numerous algorithms.
Decentralized traffic aware scheduling (DeTAS)~\cite{7254107} uses a
hierarchical
approach where all nodes follow a macro schedule that is a combination of
micro-schedules for each routing graph.
Orchestra~\cite{Duquennoy:2015:ORM:2809695.2809714} is the first algorithm
towards
autonomous scheduled TSCH where nodes compute their own schedule locally
and autonomously based on the routing layer information.

At the MAC layer, decentralized scheduling results in cell overlapping
and thus in many collisions.
A collision occurs when the same cell is allocated to different pairs in
the same interference range.
In order to avoid cell overlapping and reduce internal interference,
Decentralized Broadcast-based Scheduling algorithm
(DeBraS)~\cite{Municio:2016:DBS:2980137.2980143} allows nodes to share
scheduling information.
The collision reduction and throughput improvement by DeBraS for
dense networks come at the cost of higher energy consumption.

The algorithm proposed in~\cite{Hwang:2017:DSA:3066575.3066845} allows
every sensor node to compute its time-slot schedule in a distributed manner.
A scheme called Reliable, Efficient, Fair and Interference-Aware Congestion
Control
(REFIACC)~\cite{Kafi20171}, takes into account the heterogeneity in link
interference and capacity when constructing the scheduling send policy
in order to reach the maximum fair throughput in wireless sensor networks.
The authors in~\cite{7478619} proposed a "housekeeping" mechanism which
detects scheduled collisions and reallocates each colliding cell to a
different position in the schedule.
A distributed cell-selection algorithm for reducing scheduling errors
and collisions is proposed in~\cite{Duy201780}.
It considers a fixed queue length, thus, the algorithm cannot adapt
to the network conditions.

Wave~\cite{Soua:2016:WDS:3035337.3035346} builds the schedule by constructing
a series of waves in order to minimize the delay in convergecast applications.
All successive waves are copies of the first wave, where slots without
scheduled transmissions are removed.
An extension of Wave, where subsequent waves overlap, has been presented
in~\cite{conf/im/SouaML15}.

Recently, Decentralized Adaptive Multi-hop scheduling for 6TiSCH Networks
(DeAMON)~\cite{8017573} and Recurrent Low-Latency Scheduling Function
(ReSF)~\cite{DANEELS2018100} have been proposed.
ReSF minimizes the latency by reserving minimal-latency paths from the source to
the sink, and it only activates these paths when recurrent traffic is expected.
This results into a latency improvement of up to 80\% compared to
state-of-the-art low-latency scheduling functions. This improvement comes on
the cost of an increased power consumption.

\subsection{6TiSCH bandwidth reservation algorithms}

On-the-fly (OTF)~\cite{7273816} is a distributed algorithm that dynamically
adapts the bandwidth allocation by calculating the number of cells to be
added or removed according to a neighbor-specific threshold.
OTF is prone to schedule collisions since nodes might not be aware of
which cells are allocated to other pairs of nodes.

The authors in~\cite{8449793} assess the performance of OTF in terms of
reliability and latency.
In their assessment they focus on the impact of the network dynamics on the OTF
performance, namely the routing protocol and the 6top negotiations.
Based on the their analysis, they propose Enhanced-OTF (E-OTF) which improves
the OTF performance by modifying the allocation algorithm in OTF\@.
First, E-OTF considers the channel quality in the computation of resources by
introducing a measure for the average number of required retransmissions for a
successful packet transmission.
Second, it includes a mechanism to recover from congestion by taking into
consideration the amount of queued data.
However, E-OTF does not consider the energy consumption, that is one of the
main requirements for efficient resource allocation algorithms for IoT devices.

Scheduling Function Zero (SF0)~\cite{dujovne-6tisch-6top-sf0-01} adapts
dynamically the number of reserved cells between neighboring nodes based
on the application's bandwidth requirements and the network conditions.
SF0 uses Packet Delivery Rate (PDR) statistics to reallocate cells when
the PDR of one or more cells is much lower than the average PDR\@.
Cost-aware cell relocation (CCR)~\cite{doi:10.1002/ett.3211} complements
SF0 by detecting scheduling collisions and relocating the involved cell.
To detect collisions, CCR compares the PDR of all cells to a particular
neighbor.
If one cell has a PDR significantly lower than the PDR of other cells, then
there is a schedule collision and that cell is relocated to a different
slot/channel in the TSCH schedule.

References~\cite{DeGuglielmo20161,TELESHERMETO201784,KHARB201959} provide an
extended literature overview of scheduling algorithms in IEEE 802.15.4e.

In contrast to the related works, this paper introduces congestion control
to the scheduling algorithm, in a way that leads to optimal performance in
terms of delay.
Specifically, the local voting mechanism is used for
determining how many cells should be allocated to each link, not only
based on its own traffic requirements, as other schemes already do,
but also considering the traffic requirements and allocations of the neighboring
(conflicting) links,
This algorithms allocates the cells in a way that minimizes the maximum
value of the ratio of queue length over number of cells (the load).
Since the delay per link is given by this ratio, the allocation ensures
the minimum delay per link under congestion.

\section{Network Model and Problem Formulation} \label{model}

Our model considers a 6TiSCH network which has built a tree routing topology
with one or multiple parents per node, using the Routing Protocol for Low-Power
and Lossy Networks (RPL)~\cite{rfc6550}.
For reader's convenience, Table~\ref{notations} summarizes the notation
used throughout this paper.

  \begin{table}
    \begin{center}
            \caption{List of notations}
            \label{notations}

      \begin{tabular}{|l|p{90mm}|}
        \hline
        ${\cal G}=(V,E)$        &  Network topology graph where $V$ is
                                   the set of all nodes and $E$ is the
                                   set of edges between the nodes \\
        $N$                     &  Total number of nodes in the
                                   network \\
        ${\cal G}_T=(V_T, E_T)$ &  Tree topology graph where
                                   $V_T\subseteq V$ and
                                   $E_T\subseteq E$\\
        $n_0$                   &  Sink node, $n_0 \in V_T$ \\
        $n_i$                   &  $i$-th node in the network,
                                   $n_i \in V, 1\leq i < N$ \\
        $f$                     &  Slot frame \\ 
        $S$                     &  Total number of time slots in a slot
                                   frame \\
        $t$                     &  Time slot where $0\leq t < S$ \\
        $M$                     &  Total number of channel offsets \\
        $\emph{chOf}$           &  Channel offset where
                                   $0\leq \emph{chOf} < M$\\
        $c^{(t,
        \emph{chOf})}_{(i, j)}$ &  Cell with coordinates
                                   $(t, \emph{chOf})$ assigned to link
                                   $(i, j)$ \\
        $N_i^{(1)}$             &  Set of one-hop neighbors of node
                                   $n_i$ \\
        $N_{i,j}$               &  Set of all links that could interfere
                                   with link $(i,j)$ \\
        $q_{(i,j)}^{f}$         &  Number of packets that $n_i$ sends to
                                   $n_j$ at frame $f$ \\
        $p_{(i,j)}^{f}$         &  Number of allocated cells to link
                                   $(i, j)$ at frame $f$ \\
        $z_{(i,j)}^{f}$         &  Number of new packets received by
                                   $n_i$ with destination $n_j$ at frame
                                   $f$ \\ 
        $u_{(i,j)}^{f}$         &  Number of cells added/deleted to link
                                   $(i,j)$ at frame $f$ due to Local
                                   Voting \\ 
        $r_{(i,j)}^{f}$         &  Number of cells released from link
                                   $(i,j)$ at frame $f$ \\
        $x_{(i,j)}^{f}$         &  Load of link $(i, j)$ at frame $f$\\
        $\tilde N_{(l,k)} \subset N_i^{(1)} $  &  Set of neighboring links of link $(l,k)$, that can give at least one cell to link $(l,k)$\\
        \hline
      \end{tabular}
    \end{center}
  \end{table}

The communication in the network can be modeled by a graph ${\cal G}=(V,E)$,
where $ V=\{n_i:  0 \leq i < N\} $ is the set of all nodes and $E$ is the
set of edges that represent the communication symmetric links between the nodes.
Data is gathered over a tree structure ${\cal G}_T=(V_T,E_T)$ rooted at the
sink node $n_0$ where $n_0\in V_T, V_T\subseteq V$, and $E_T\subseteq E$.
We consider both of the cases where each node has only one parent (tree) and
where there are multiple parents per node.
Without loss of generality, we consider a single-sink model.
We assume that all nodes are synchronized, and each node has a single
half-duplex radio transceiver.
Since the communication is half-duplex, each node cannot transmit and
receive simultaneously on the same channel.
We propose a {\it{link scheduling}} algorithm where a link $(i, j)$ is
a pairwise assignment of a directed communication between a pair of nodes
$(n_i, n_j)$, where $i\neq j$, in a specific time slot within a given
frame and a channel.

Time in TSCH is slotted, and assumed to be (almost) perfectly synchronized
in the whole system.
The basic time interval is referred to as a time slot $t$.
A time slot $t$ is long enough for one packet to be sent from node $n_i$
to node $n_j$ and optionally node $n_j$ to reply.
This is represented in Fig.~\ref{schedule} for the node pair $(n_3, n_1)$.
Each frame $f$ consists of equal number of $S$ time slots with the same
duration $f=\{0,\ldots, S-1\}$.
The resource allocation in a 6TiSCH network is controlled by a TSCH schedule
that allocates cells for node communication.
A cell represents a unit of bandwidth that is allocated based on a decision
by a centralized or a distributed scheduling algorithm.
As explained previously, a cell is defined by a pair of time slot and
channel offset~\cite{watteyne:hal-01208395}.
The slot offset is equal to time slot $t$ while the channel offset
\emph{chOf} is translated into a frequency with the following equation:
\begin{equation}
      \label{freq}
      channel = F\left((chOf + ASN) \mod M\right),
\end{equation}
where \emph{chOf} denotes the channel offset, $ASN$ counts the number
of time slots since the network started, $M$ is the number of physical
channels (by default 16 in TSCH), and $\mod$ is the modulo operator.
$F\left( \cdot \right)$ is a bijective function mapping an
integer comprised between 0 and 15 into a physical channel.
The number of channel offsets is equal to the number of available
frequencies $0\leq \emph{chOf} < M$.
The schedule can be represented by a matrix with dimensions: the total
number of channel offsets $M$ and total number of time slots in a
slot frame $S$. One example of a schedule with 4 time slots and 3 channels is
given in Fig.~\ref{schedule}.
Note that some cells can be shared between different links (e.g. $n_8
\rightarrow n_5$ and $n_6 \rightarrow n_4$), as long as there is no
mutual interference.
A TSCH schedule instructs node $n_i$ what to do in a specific time slot
and frequency: transmit, receive, or sleep.
The cell assigned to link $(i, j)$ at slot offset $t$ and channel offset
$\emph{chOf}$ is denoted by $c^{(t, \emph{chOf})}_{(i, j)}$ where
  \begin{equation}
      \label{cell}
      c^{(t, \emph{chOf})}_{(i, j)} = \left\{
      \begin{array}{rl}
          1, & n_i \mbox{ transmits and } n_j \mbox{ receives at } t
               \mbox{ and } \emph{chOf}; \\
          0, & \mbox{otherwise};
      \end{array} \right.
  \end{equation}
for $n_i\in V, 0\leq t\leq S-1,$ and $0\leq \emph{chOf} \leq M-1$.

There exists a scheduled cell for node $n_j$ from the pair $(n_i, n_j)$ such
that $n_j$ receives the transmission from $n_i$ at the same $t$ and
$\emph{chOf}$ that are scheduled for transmission of node $n_i$.
Each scheduled cell is an opportunity for node $n_i$ to communicate with its
one-hop neighbor $n_j$ where $n_j \in N_i^{(1)}$ and $N_i^{(1)}$ denotes
the one-hop neighborhood of node $n_i$.
We consider an interference model where two nodes are one-hop neighbors as
long as their Packet Delivery Rate (PDR) is larger than 0.
In real scenarios, nodes that are at more than two hops distance can also
interfere but with lower probability~\cite{journals/ett/MorellVVW13}, thus,
we only consider the one-hop neighborhood.

  \begin{figure}
    \centering
    \includegraphics[width=\textwidth]{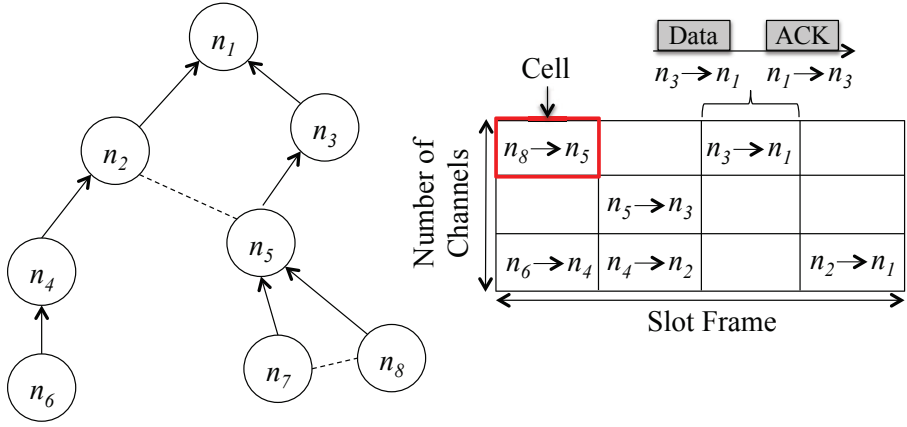}
    \caption{TSCH schedule for the presented topology where solid lines
             represent connection between nodes based on RPL and dashed
             lines represent possible communication between nodes.}
    \label{schedule}
  \end{figure}

The 6top sublayer qualifies each cell as either a hard or a soft cell.
A soft cell can be read, added, removed, or updated by the 6top sublayer,
while a hard cell is read-only for the 6top sublayer.
In the context of the proposed algorithm, all reallocated cells are soft
cells.

The role of the bandwidth reservation algorithm is to ensure that there
are enough resources to meet the application requirements such as traffic load,
end-to-end delay, and
reliability.
The proposed scheduling algorithm must satisfy the following communication
conditions:
\begin{enumerate}
  \item Multi-point to point communication where data is generated only
        by source nodes $n_i$, where $n_i \in V_T$, and it is gathered at
        the sink node $n_0$.
  \item The communication is half-duplex, thus, each node cannot transmit
        and receive simultaneously on the same channel.
  \item Nodes $n_i$ and $n_j$ from the pair $(n_i, n_j)$ transmit and
        receive in the same cell, i.e., $(t, \emph{chOf})$, respectively.
  \item \textit{Collision-free} communication: A cell with coordinates
        $(t, \emph{chOf})$ is allocated to link $(i, j)$ such that exactly
        one of the neighbors, i.e., node $n_i$, of the receiving node $n_j$
        should transmit at slot offset $t$ and channel offset
        $\emph{chOf}$, and the other neighbors $n_l$ of the receiving
        node $n_j$, where $n_l \in N_{j}^{(1)}$ and $n_l\neq n_i$, might
        receive at slot offset $t$ and channel offset $\emph{chOf}$.
\end{enumerate}
In general, to prevent collisions between pairs of links $(i, j)$ and
$(l, k)$, the following \textit{collision-free} constraints are defined.

\emph{Primary conflict constraint:} A node cannot transmit and/or
receive two packets at the same time slot $t$, even not on different
channels \emph{chOf1} and \emph{chOf2}, i.e.,
  \begin{equation} \label{collision-free1}
  \begin{split}
      c^{(t, \emph{chOf1})}_{(i, j)} c^{(t, \emph{chOf2})}_{(l, k)} = 0,\mbox{for all:}
      \{i,j\}\cap\{k,l\}\neq \emptyset, \\ n_k \in N_{i}^{(1)}, n_l \in N_j^{(1)}.
      \end{split}
  \end{equation}

Eq.~(\ref{collision-free1}) indicates that the communication is half-duplex.

\emph{Secondary conflict constraint:} A receiver cannot decode an incoming
packet in a channel \emph{chOf}, if another node in its neighborhood is
also transmitting at the same channel \emph{chOf} at the same time slot $t$.
Hence, a node is not allowed to receive more than one transmission
simultaneously, i.e.,
  \begin{equation} \label{collision-free2}
      c^{(t, \emph{chOf})}_{(i, j)} c^{(t, \emph{chOf})}_{(l, k)} = 0, \mbox{for all:}
      n_k \in N_{i}^{(1)}, n_l \in N_{j}^{(1)}.
  \end{equation}

Eq.~(\ref{collision-free2}) indicates the interference constraint.

\section{Local Voting Bandwidth Reservation Algorithm}
\label{lv}

Each source node $n_i$, where $n_i \in V_T$ and $n_i\neq n_0$, has a queue
with packets to be transmitted to the root through a parent node,
which is a one-hop neighbor of the node $n_i$.
The internal scheduling on the queue is first-come-first-serve.
A cell is allocated to link $(i, j)$ so that node $n_i$ transmits a packet
to $n_j$ as it is given in Eq.~(\ref{cell}).

The state of link $(i, j)$, where $n_j \in N_i^{(1)}$, at
the beginning of frame $f+1$ is described by three characteristics:
\begin{itemize}
  \item $q_{(i, j)}^{f+1}$ is the number of packets (queue length) that
        node $n_i$ has to transmit to node $n_j$ at slot frame $f+1$;
  \item $p_{(i, j)}^{f}$ is the number of cells allocated to link
        $(i, j)$ at the previous slot frame $f$, i.e.,
        $p_{(i, j)}^{f} =
                \sum\limits_{t=0}^{S-1}  c_{(i, j)}^{(t, \emph{chOf})}$.
\end{itemize}
There is no sum over the channels in the equation for calculating
$p_{(i, j)}^{f}$ due to the fact that each node has a single transceiver,
so each link can be allocated only one channel at each time slot.

The dynamics of each link $(i, j)$ are calculated as:
\begin{eqnarray}
   \label{dynamics}
   \begin{aligned}
      q_{(i, j)}^{f+1} &= \max\{0, q_{(i, j)}^{f}
                     - p_{(i, j)}^{f+1} \} + z_{(i, j)}^{f},  \\
      p_{(i, j)}^{f+1} &= p_{(i, j)}^{f} + u_{(i, j)}^{f+1},
  \end{aligned}
\end{eqnarray}
where
\begin{itemize}
  \item $z_{(i, j)}^{f}$ is the number of new packets received from upper
        layers or from neighboring nodes of node $n_i$ with a next-hop
        destination equal to node $n_j$ at frame $f$;
  \item $u_{(i, j)}^{f+1}$ is the number of cells that are added or
        released to link $(i, j)$ at frame $f+1$ due to LV.
\end{itemize}

The objective of the proposed LV algorithm is to schedule link transmissions
in such a way that the minimum maximal (min-max) link delay is achieved.
The algorithm stems from the finding that the shortest delivery time is
obtained when the load is equalized throughout the network.
The finding has been proved in~\cite{7047923} for the case of load balancing
in cluster computing, while~\cite{8091101} presents a similar result for the
case of wireless mesh networks with a single channel
and a node scheduling MAC layer.
In this paper we extend the result for link scheduling where the MAC layer
follows TSCH mode.

The load of link $(i, j)$ at frame $f$ is defined as the ratio of the queue
length $q_{(i, j)}^{f}$ over the number of allocated cells $p_{(i, j)}^{f}$
as follows:
\begin{equation}
  \label{load}
  x_{(i, j)}^{f} = \left\{
    \begin{array}{rl}
      \left[ \cfrac{q_{(i, j)}^{f}}{p_{(i, j)}^{f}} +0.5 \right],
                                      & \mbox{ if } q_{(i, j)}^{f} > 0, \\
          0,                          & \mbox{ if } q_{(i, j)}^{f}= 0,
    \end{array} \right.
\end{equation}
where
$\left[ \cdot \right]$ is the round function
(rounds a real number to the nearest integer).
Note that by construction of the above definition, the delay at each link
$(i,j)$ (in time slots) can be computed as $x_{(i,j)} \cdot |S|$, where $S$
is the number of slots in a slot frame.

In order to semi-equalize or balance the load in the network, neighboring
links can exchange cells as long as Eq.~(\ref{collision-free1}) and
Eq.~(\ref{collision-free2}) are satisfied. The set $N_{i,j}$ contains all
links that could potentially interfere with link $(i,j)$.
This means that
$$
  (l,k) \in N_{i,j} \mbox{ iff } n_k \in N^{(1)}_i \lor n_l \in N^{(1)}_j.
$$

\begin{definition} A conflict-free schedule is link-wise optimal
  or just optimal, if the maximum delay per link in the network is
  smaller or equal than the maximum delay per link for every other
  schedule (min-max).
\end{definition}

\begin{lemma}
  (Optimal schedules are maximal) An optimal schedule is a
  (or has an equivalent) maximal schedule in the sense that\footnote{%
  Symbol $\nexists$ denotes the negation
  of existence $\exists$}
  $\nexists (i,j) \in E $ such that $p_{(i,j)}$ can be increased without
  reducing
  $p_{(l,k)}$ for at least one other link where $(l,k) \in E$.
\end{lemma}

\begin{proof} Consider a schedule that is not maximal. That means there
  exists $(i,j) \in E$ such that $p_{(i,j)}$ can be increased by one.
  Cells are not reallocated to other links, therefore, for the new
  schedule, the delay for all the other links is unchanged.
  For link $(i,j)$, the new delay is $x'_{(i,j)} \cdot |S| = \left[
  \frac{q_{(i,j)}}{(p_{(i,j)} + 1)} +0.5 \right] \cdot |S| \leq
  x_{(i,j)} \cdot  |S|$.
  It follows that, for every
  non-maximal schedule, there exists a maximal schedule that has smaller
  or equal maximum delay.
\end{proof}

\begin{lemma}\label{balanced}(Optimal schedules are balanced) Assume that
  link $(l,k)$ is the most loaded link in the network, i.e.,
  $(l,k) = {\rm argmax}(x_{(i,j)}), (i,j) \in E$.
  For all optimal schedules, it holds that
  $x_{(l,k)} \leq x_{(i,j)} / (1-1/p_{(i,j)})$
  for the load of the most loaded link $(l,k)$ and the load of every other link
  $(i,j)$, where $(i,j)\in \tilde N_{(l,k)}$, and $ \tilde N_{(l,k)} $ is the
  set of neighboring links of link $(l,k)$ that can give at least one cell to
  link $(l,k)$, without violating the constraints from
  Eq.~(\ref{collision-free1}) and Eq.~(\ref{collision-free2}).
\end{lemma}

\begin{proof} Assume that an optimal schedule exists where for the most
  loaded link $(l,k)$, $x_{(l,k)} > x_{(i,j)} / (1-1/p_{(i,j)})$ where ${(i,j)}
  \in \tilde N_{(l,k)}$.
  Since $(l,k)$ is the most loaded link, the maximal delay for such a schedule
  is $x_{(l,k)} \cdot  |S|$.
  Since link $(i,j) \in \tilde N_{(l,k)}$, it follows that a cell of link
  $(i,j)$ can be reassigned to link $(l,k)$.
  After reassigning, the new load for link $(l,k)$ is
  $ x'_{(l,k)} =[ q_{(l,k)}/(p_{(l,k)} +1) +0.5]$,
  and the corresponding delay for link $(l,k)$ is
  $[ q_{(l,k)}/(p_{(l,k)} +1) + 0.5 ] \cdot  |S| < x_{(l,k)} \cdot  |S|$.
  In addition, link $(i,j)$ loses a cell so the new delay for link $(i,j)$ is
  $[ q_{(i,j)}/(p_{(i,j)} - 1) +0.5 ] \cdot  |S|
        = [ (q_{(i,j)}/p_{(i,j)})/(1 - 1/p_{(i,j)}) +0.5 ] \cdot  |S|
        = [(q_{(i,j)}/p_{(i,j)})+0.5] /(1 - 1/p_{(i,j)}) \cdot  |S|
        = x_{(i,j)}/(1 - 1/p_{(i,j)}) \cdot  |S|
        < x_{(i,j)} \cdot  |S|$.
  Thus, the new allocation has a maximal delay that is smaller than or equal to
  the maximal delay of the other allocation, so the allocation is not optimal.
\end{proof}

Based on the above reasoning, we design a load balancing strategy with two
goals:
\begin{enumerate}
  \item the produced schedule should be maximal; and
  \item the load in the schedule should be balanced.
\end{enumerate}

It should be noted that, in general, a schedule could be both maximal and
balanced, but still not optimal.
This is because there could exist a reallocation of the slots in the network
that would produce a larger spectral efficiency.
Optimizing the schedule in this sense would require finding a solution for the
{\it NP}-complete link scheduling problem.
This is not easy, so for the purposes of this paper, we do not examine ways of
escaping local optima and finding the global optimum.
However, the simulation results show that the performance of LV is still better
than the performance of the algorithms that we compare with,
and that optimizing the maximal nodal delay also has a positive impact
on the end-to-end delay.

In the following part we explain LV and the way how $u_{(i, j)}^{f+1}$ is
calculated.
LV triggers the 6top sublayer to add and release cells to link $(i, j)$ at
frame $f+1$ for $u_{(i, j)}^{f+1}>0$ and $u_{(i, j)}^{f+1}<0$, respectively.
The value of $u_{(i, j)}^{f+1}$ is calculated as:
\begin{equation} \label{u}
  \begin{split}
    {u}_{(i, j)}^{f+1} &= \\
    &\left[ \cfrac{\left(q_{(i, j)}^{f} + z_{(i, j)}^{f+1} \right)\times S}{%
                      q_{(i, j)}^{f}\!+\!z_{(i, j)}^{f+1}\!%
                      +\!\sum\limits_{(l,k)\!\in\!N_{i,j}}{%
                             \!w_{(i,j,l,k)}\!\left(\!q_{(l, k)}^{f}\!%
                             +\!z_{(l, k)}^{f+1}\!\right) }}
          \right] \\& - p_{(i, j)}^{f},
  \end{split}
\end{equation}
where
\begin{equation}\label{w}
  w_{(i,j,l,k)} = \left\{
  \begin{array}{rl}
     1,            & \mbox{ if } \{i,j\}\cap\{k,l\}\neq \emptyset, \\
     \sfrac{1}{M}, & \mbox{ otherwise. }
  \end{array} \right.
\end{equation}

The value in the round function in Eq.~(\ref{u}) is the number of cells
allocated to link $(i,j)$ at frame $f$ taking into account the new packets from
upper layers or neighboring nodes.
As we can see from the term $q_{(i, j)}^{f}$, the number of allocated
cells is proportional to the queue length within the neighborhood of link
$(i,j)$, so it leads to semi-equal load between the neighboring links.
Also, we scale to the total number of time slots that are needed to transmit
all queued packets in the neighborhood of
link $(i,j)$, so that the total number of time slots in the neighborhood is
equal to
the number of time slots in the frame.
The weight $w_{(i,j,l,k)}$ is used to capture the difference between a
primary and a secondary conflict.
In the first case, since all channels are unavailable to the link, the
value is one, but in the second case, since only one of the available
channels is blocked, the value is $1/M$.

The rationale of Eq.~(\ref{u}) can be also seen if we calculate the load at the
end of frame $f+1$.
If $q_{(i, j)}^{f}> p_{(i, j)}^{f+1} $, then we have
$ x_{(i,j)}^{f+1} = \frac{q_{(i, j)}^{f+1}}{p_{(i, j)}^{f+1}}
  = \frac{q_{(i, j)}^{f} - p_{(i, j)}^{f+1}
    + z_{(i, j)}^{f+1}}{p_{(i, j)}^{f+1}}
  = \frac{q_{(i, j)}^{f} + z_{(i, j)}}{p_{(i, j)}^{f+1}} - 1
$.
In addition, we have that
$$ p_{(i, j)}^{f+1}
   = p_{(i, j)}^{f} + u_{(i, j)}^{f+1}
   = \cfrac{\left(q_{(i, j)}^{f} + z_{(i, j)}^{f+1} \right)\times S}{
                      q_{(i, j)}^{f}\!+\!z_{(i, j)}^{f+1}\!%
                      +\!\sum\limits_{(l,k)\!\in\!N_{i,j}}{%
                             \!w_{(i,j,l,k)}\!\left(\!q_{(l, k)}^{f}\!%
                             +\!z_{(l, k)}^{f+1}\!\right) }},$$
which means that
$$ x_{(i,j)}^{f+1} = \cfrac{
                      q_{(i, j)}^{f}\!+\!z_{(i, j)}^{f+1}\!%
                      +\!\sum\limits_{(l,k)\!\in\!N_{i,j}}{%
                             \!w_{(i,j,l,k)}\!\left(\!q_{(l, k)}^{f}\!%
                             +\!z_{(l, k)}^{f+1}\!\right) }}{S} - 1.$$

We will show that this quantity is invariant for the links $(i,j)$ and $(j,k)$
that share the same neighborhood.
  For $\{i,j\}\cap\{k,l\} = \emptyset$, we have
  $$ x_{(i,j)}^{f+1} =  \cfrac{
                      q_{(i, j)}^{f}\!+\!z_{(i, j)}^{f+1}\!%
                      +\!\sum\limits_{(l,k)\!\in\!N_{i,j}}{%
                             \!w_{(i,j,l,k)}\!\left(\!q_{(l, k)}^{f}\!%
                             +\!z_{(l, k)}^{f+1}\!\right) }}{S} - 1.$$

  By substituting~(\ref{dynamics}) into~(\ref{u}), we get
\begin{equation}
  \label{u2}
  \begin{aligned}
      {u}_{(i, j)}^{f+1} = &\\
          & \left[ \cfrac{\left(\max\{0, q_{(i, j)}^{f} - p_{(i, j)}^{f+1} \}
          + z_{(i, j)}^{f} \right) \times S}{\splitdfrac{%
          \left(\max\{0, q_{(i, j)}^{f} - p_{(i, j)}^{f+1} \}
          + z_{(i, j)}^{f} \right)}
          {+ \sum_{(l,k) \in N_{i,j}}{w_{(i,j,l,k)} \times \left(\max\{0,
          q_{(l, k)}^{f} - p_{(l, k)}^{f+1} \} + z_{(l, k)}^{f} \right)}}}
          \right] \\ - &p_{(i, j)}^{f}.
         \end{aligned}
\end{equation}

\subsection{The Local Voting Algorithm}
Alg.~1 presents Local Voting algorithm.
All links (edges) are examined sequentially at the beginning of each frame.
The source node requests for cells, not the receiver. Since we consider a
link scheduling scenario, the destination of each transmission is known
during the scheduling phase.
Every link in the network that has a positive queue length calculates a
value $u^{f+1}$ (given in Eq.~(\ref{u})).
If node $n_i$ has packets to send to node $n_j$, the value of
${u}_{(i, j)}^{f+1}$ determines the number of cells that the link $(i,j)$ should
ideally gain or release at slot frame $f+1$.
If $u_{(i,j)}^{f+1}$ is a positive value, then LV asks from the 6top
sublayer to add cells to link $(i,j)$.
Otherwise, if $u_{(i,j)}^{f+1}$ is a negative value, then LV requests from
the 6top sublayer to release $u_{(i,j)}^{f+1}$ cells that have been
allocated to $(i,j)$.
The cell reallocation should not cause collisions with respect to
Eq.~(\ref{collision-free1}) and Eq. (\ref{collision-free2}).
The collision-free constraint is implemented in 6top sublayer which is
responsible for collision-free communication.
On the other hand, if node $n_i$ does not have packets to send to
destination $n_j$ and cells have been already allocated to link $(i,j)$ in
the previous frame, then all allocated cells $p_{(i,j)}^{f}$ are released.
In general, cells are removed from links with a lower load and are offered
to links with a higher load.

\begin{algorithm} \label{alg2}
  \floatname{algorithm}{Algorithm}
  \begin{algorithmic}
    \For {$(i, j) \in E$}   \Comment{Check for all outgoing links
                                $(i,j)$ that originate at node $n_i$}
    \State{$qsum_{(i,j)}^{f+1} =
        (q_{(i, j)}^{f} + z_{(i, j)}^{f+1})  +\sum_{(l,k) \in N_{i,j}}{w_{(i,j,l,k)}
        \times (q_{(l, k)}^{f}} + z_{(l, k)}^{f+1})$}
    \If {$qsum_{(i,j)}^{f+1} \neq 0$} \Comment {Are there packets in the
                               neighborhood of link $(i,j)$ to be sent?}
    \State Calculate $u_{(i,j)}^{f+1} =
        \left[ \frac{(q_{(i, j)}^{f}+z_{(i, j)}^{f+1}) \times S}{qsum_{(i,j)}^{f+1}}
        \right] - p_{(i, j)}^{f}$
    \If {$u_{(i,j)}^{f+1} > 0$} \Comment{The link requests cells}
    \State{Request from 6top to add $u_{(i,j)}^{f+1}$ cells to
        link~$(i,j)$}
    \ElsIf{$u_{(i,j)}^{f+1} < 0$} \Comment{The link releases cells}
    \State{Request from 6top to delete $u_{(i,j)}^{f+1}$ cells from
           link~$(i,j)$}
    \EndIf
    \ElsIf{$p_{(i,j)}^{f} \neq 0$} \Comment{Are there cells allocated
                                        to a link with an empty queue?}
    \State Request from 6top to delete $p_{(i,j)}^{f}$ cells from
           link~$(i,j)$ \Comment{Release the allocated cells}
    \EndIf
    \EndFor
  \end{algorithmic}
  \caption{Local Voting}
\end{algorithm}
To summarize, LV requests from the 6top sublayer to add cells to link
$(i,j)$ at slot frame $f+1$ when:
\begin{itemize}
  \item node $n_i$ has packets to send to node $n_j$ and the value of
        $u_{(i,j)}^{f+1}$ for link $(i,j)$ is positive which means that
        the link $(i,j)$ has a higher load than its neighbors.
\end{itemize}
LV requests from the 6top sublayer to release cells from link $(i,j)$ at
slot frame $f+1$ when:
\begin{itemize}
  \item node $n_i$ has packets to send to node $n_j$ and the value of
        $u_{(i,j)}^{f+1}$ is negative which means that the link $(i,j)$
        has a lower load than its neighbors; or
  \item node $n_i$ does not have packets to send to node $n_j$ and cells
        have been already allocated to link $(i,j)$.
\end{itemize}

\section{Performance Evaluation}
\label{perf}

  \nop{%
  \subsection{Local Voting-6top Interaction}

  Local Voting is located on the top of the 6top sublayer.
  As OTF algorithm~\cite{7273816}, Local Voting issues two 6top commands:
  create and delete soft cell ($\emph{CMD\_ADD}$ and
  $\emph{CMD\_DELETE}$~\cite{wang-6tisch-6top-sublayer-04}).
  It retrieves statistics from the 6top sublayer about the list of neighboring
  links, the queue length of each link, and the number of scheduled cells per
  link.
  Not all of these statistics were available in the reference implementation
  so the 6top sublayer had to be modified to accommodate the additional
  parameters for LV.

    \subsection{Simulations with the 6TiSCH Simulator}
    \label{evaluation}}

The 6TiSCH simulator is an open-source, event-driven Python simulator developed
by the members of the 6TiSCH WG~\cite{simul}.
Reference~\cite{municio2018simulating} discusses the overall architecture of
the 6TiSCH simulator, its use for simulating realistic scenarios, and published
results that use the 6TiSCH simulator for different purposes.
By default, the simulator supports IEEE 802.15.4e TSCH mode~\cite{MACstandard},
RPL~\cite{rfc6550}, 6top~\cite{wang-6tisch-6top-sublayer-04}, and
OTF~\cite{7273816}.
In addition to these protocols, we have added Local Voting and Enhanced OTF
(E-OTF) algorithms\footnote{%
  As an online addition to this article, the source code is
  available at \url{https://github.com/djvergad/local_voting_tsch}
} as part of the work presented in this article.
We have implemented two distinct versions of Local Voting, the original
version presented in~\cite{8480306}, and the modified version presented
in this paper.
The new version, which is marked as "local\_voting\_z" in the figures, differs
from the original by considering the new packets that are expected to arrive
at each slot frame, and not only the current state of the queues of each link
as presented with Eq.~(\ref{u}).

We compare the two versions of LV with OTF~\cite{7273816} and
E-OTF~\cite{8449793} for two threshold values, 4 and 10 cells.
We work with the same simulation parameters as in~\cite{7273816} which have
been set according to RFC5673~\cite{rfc}. The parameters are set according to
a) an industrial environment scenario where traffic can be bursty
and b) a senario where traffic is generated at a steady rate.

For the first case, consider a scenario where a leakage is detected
in an oil and gas system, the
sensors transmit at a higher sample rate in order to minimize the time for
detection of the leakage location, to calculate its magnitude, and to
estimate the impact and the evolution of the leakage.
The simulation parameters are summarized in Table~\ref{parameters}.

The simulation scenario considers a network with a grid topology of
$2km \times 2km$ where $50$ nodes are placed randomly.
Every link is associated with a Packet Delivery Rate (PDR) value between 0.00
and 1.00.
The PDR value per link is constant during a simulation run. Each node has
at least three neighbors where the PDR of the links is at least 50\%.
A node is moved until this condition is satisfied.
The minimum acceptable RSSI value that allows for a packet reception is
$-97dBm$, while the maximum number of MAC retries is set to 5.
The TSCH schedule contains 101 cells where each time slot has duration of
$10ms$.

In the bursty scenario, nodes start to generate data at $20s$ after
the beginning of the simulation.
Since data is generated in bursts, then the next data generation is at $60s$.
We perform simulations where each node
generates 1, 5 or 25 packets per burst.

In the steady state scenario the nodes started transmitting at a random
time between 16.9 and 33 seconds, and they send at an interval of
0.1, 0.2, or 0.4 seconds, with a uniform random variance of 0.05 times the interval.

The queue length of all nodes
is 100 packets.
The presented results are averaged over 500 simulation runs. A new topology
is used for each run.

\begin{table}[!t]
  \renewcommand{\arraystretch}{1.3}
  \caption{Simulation Setup}
  \label{parameters}
  \centering
  \begin{tabular}{l  p{45mm}}
    \hline
    \bfseries Parameter        & \bfseries Value \\
    \hline
    Number of Nodes            & 50 \\
    Deployment area            & square, $2km \times 2km$ \\
    Deployment constraint      & 3 neighbors with PDR>50\% \\
    Radio sensitivity          & $-97dBm$ \\
    Max. MAC retries           & 5 \\
    Length of a slot frame     & 101 cells \\
    Time slot duration         & $10ms$ \\
    Number of channels         & 16 \\
    Burst timestamp            & $20s$ and $60s$ \\
    Number of packets per burst & 1, 5, 25, 50, and 80 packets
                                  per node per burst \\
    Packet inter-arrival interval & 0.1, 0.2, and 0.4 seconds \\
    Queue length               & 100 packets \\
    Number of runs per sample  & 500 \\
    Number of cycles per run   & 100 \\
    6top housekeeping period   & $1s$  \\
    OTF threshold              & 4, 10 cells \\
    OTF housekeeping period    & $1s$ \\
    RPL parents                & 3 \\
    \hline
  \end{tabular}
\end{table}
The metrics used for performance comparison between LV, OTF and E-OTF are as
follows:
\begin{itemize}
  \item the timestamp for the last packet to reach the root shows the time
        needed for the two bursts to be completely received by the root;
  \item the end-to-end latency, defined as the time from a packet
        generation until its reception at the sink;
  \item the energy consumption, calculated by adding the energy of each
        transmission/reception/idle listen; and
  \item end-to-end reliability, defined as the ratio between the number of
        packets received by the sink and the total number of packets sent by
        all nodes;
\end{itemize}

Additionally, in order to better explain the evolution of the simulation and to
give insights on the reasons for the different performance between the
algorithms, we also depict the following values:
\begin{itemize}
  \item the evolution of the queue fill, defined as the total number of packets
        in all buffers in the system; and
  \item the distribution of load in the system, measured using Jain's fairness
        index and the G-fairness index~\cite{wiki:Fairness_measure}.
\end{itemize}

\subsection{Bursty traffic experiments}

\begin{figure}
  \centering
  \includegraphics[width=0.9\textwidth]{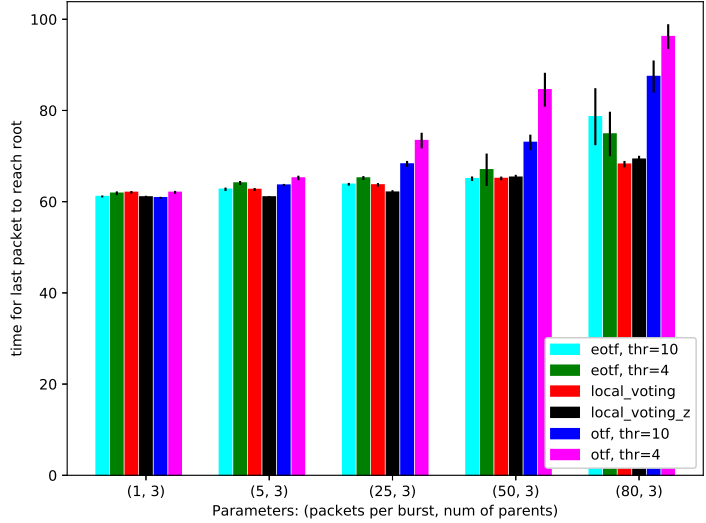}
  \caption{Time for last packet to reach root.}
  \label{time_all_root_vs_threshold_buf_100}
\end{figure}

In Fig.~\ref{time_all_root_vs_threshold_buf_100} we can see the timestamp of
the last packet that was received for each algorithm and each scenario.
In all cases Local Voting and Local Voting z deliver the packets for a shorter
time than the other algorithms, with E-OTF achieving the next best performance,
and OTF having the worst delay.

Similar results are presented in Fig.~\ref{max_latency_vs_threshold_buf_100},
where the maximum end-to-end latency is depicted for each algorithm and each
scenario.
Here we can see that Local Voting and E-OTF have similar performance for
scenarios where the load is low, but as the load increases, the advantage of
the Local Voting algorithm becomes more apparent.

Regarding the average end-to-end delay
(Fig.~\ref{latency_vs_threshold_buf_100}), again Local Voting has the best
performance, though the difference is not as prominent as in the previous
graphs.

In Figs.~\ref{appReachesDagroot_cum_vs_time_buf_100_par_3_pkt_80}--%
\ref{appReachesDagroot_cum_vs_time_buf_100_par_3_pkt_5} we can the evolution
over time of the number of packets
that have reached the root.
In all cases the black and red lines are above the other ones, that indicates that with
Local Voting a larger number of packets have been received at each timestamp.
The largest difference appears in
Fig.~\ref{appReachesDagroot_cum_vs_time_buf_100_par_3_pkt_80},
due to the higher traffic load, while the smallest difference is in
Fig.~\ref{appReachesDagroot_cum_vs_time_buf_100_par_3_pkt_5}.

\begin{figure}
  \centering
  \includegraphics[width=0.9\textwidth]{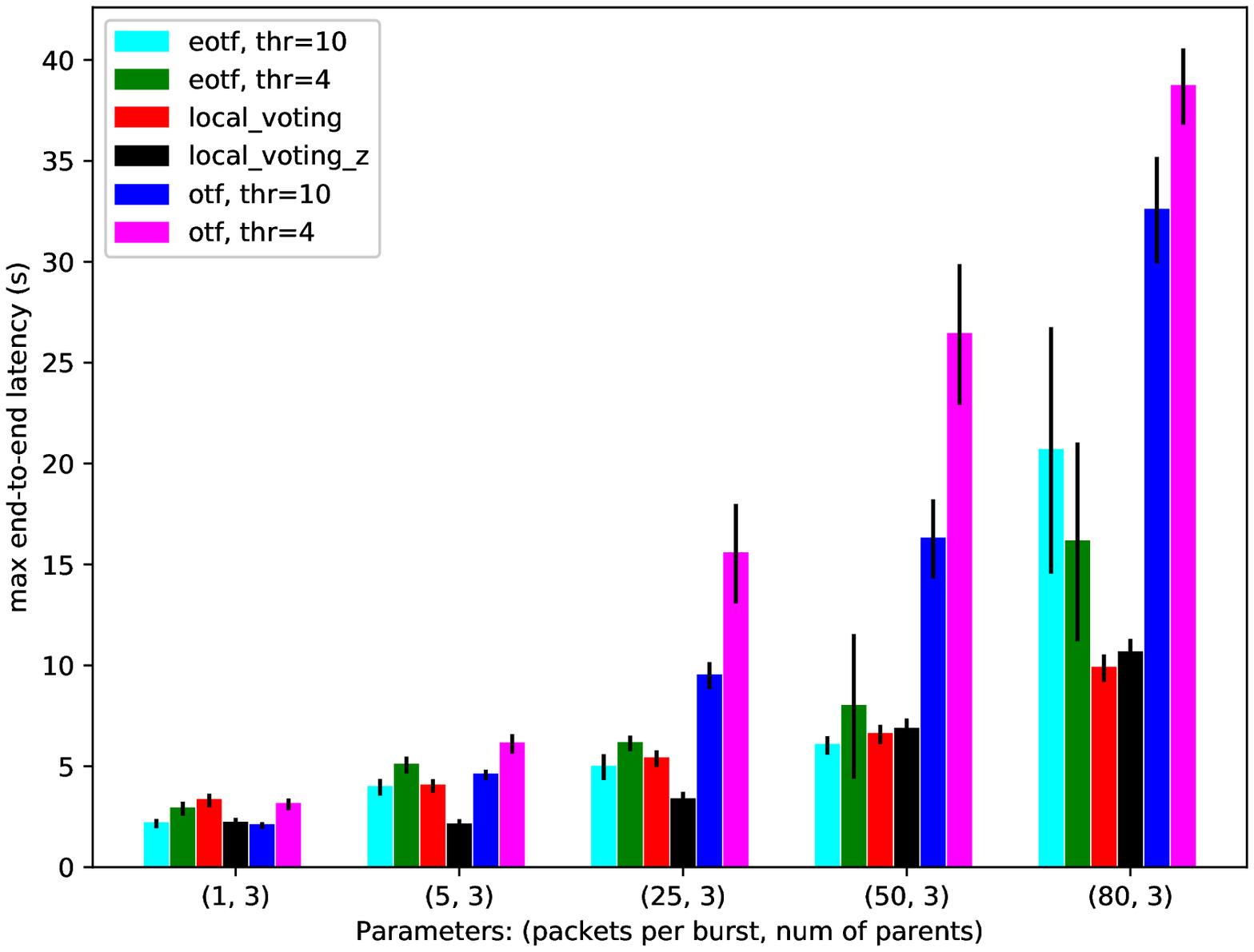}
  \caption{Maximum end-to-end latency.}
  \label{max_latency_vs_threshold_buf_100}
\end{figure}

\begin{figure}
  \centering
  \includegraphics[width=0.9\textwidth]{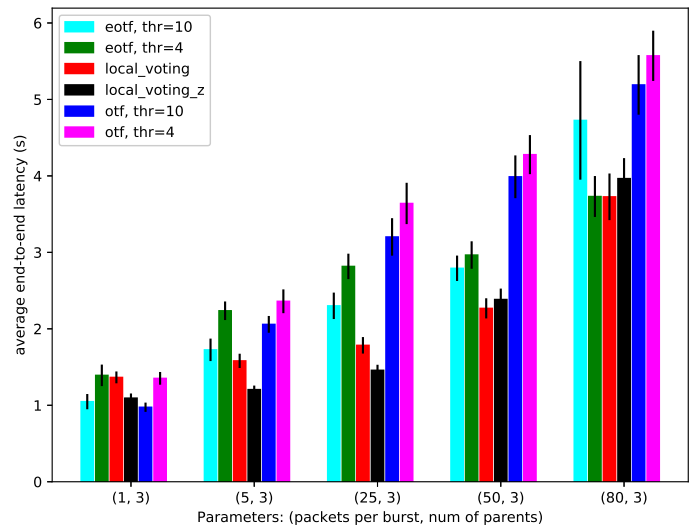}
  \caption{Average end-to-end latency.}
  \label{latency_vs_threshold_buf_100}
\end{figure}

\begin{figure}
  \centering
  \includegraphics[width=0.9\textwidth]{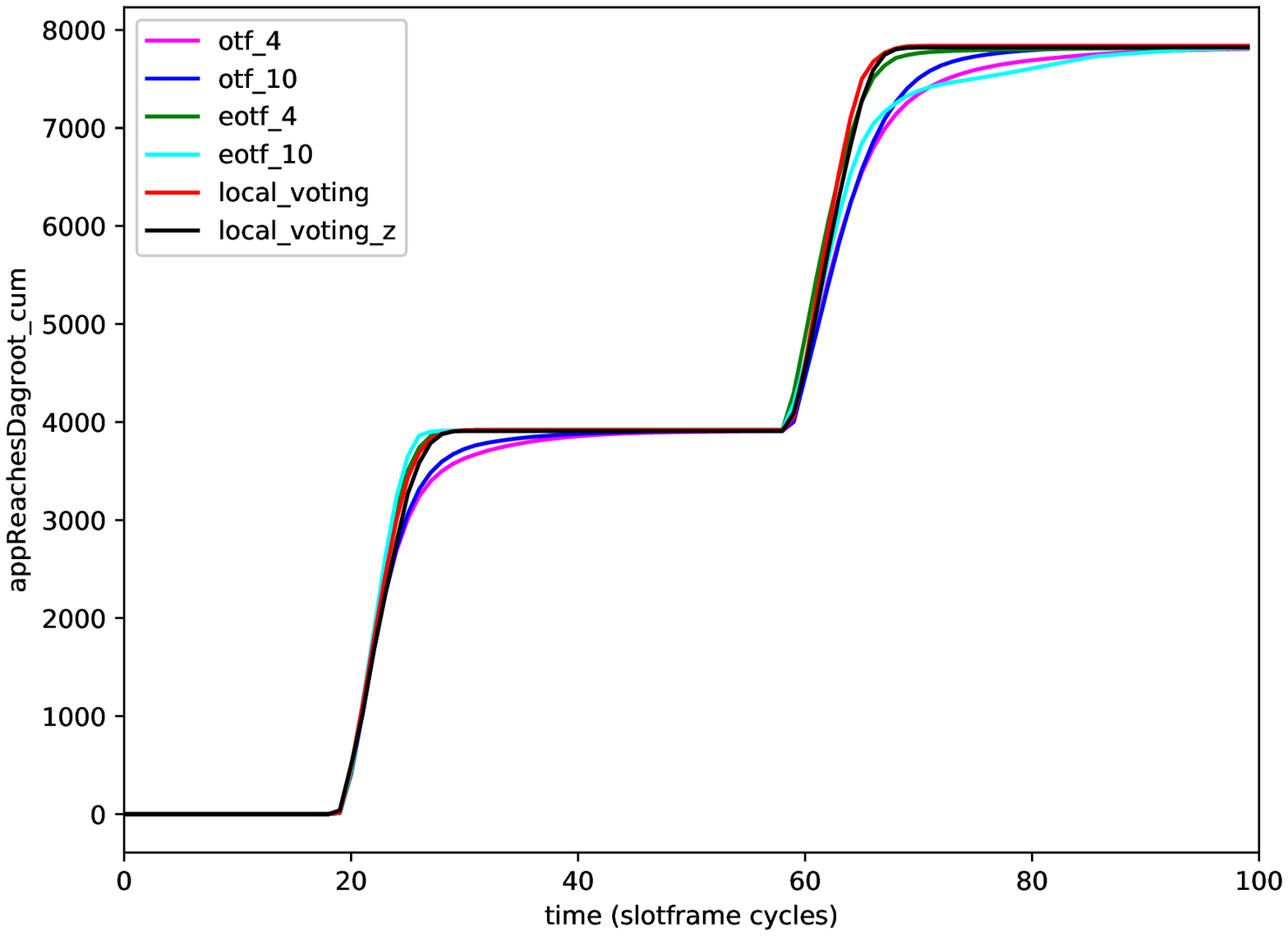}
  \caption{Number of packets that reach the root as a function of time, 80 packets per burst.}
  \label{appReachesDagroot_cum_vs_time_buf_100_par_3_pkt_80}
\end{figure}

\begin{figure}
  \centering
  \includegraphics[width=0.9\textwidth]{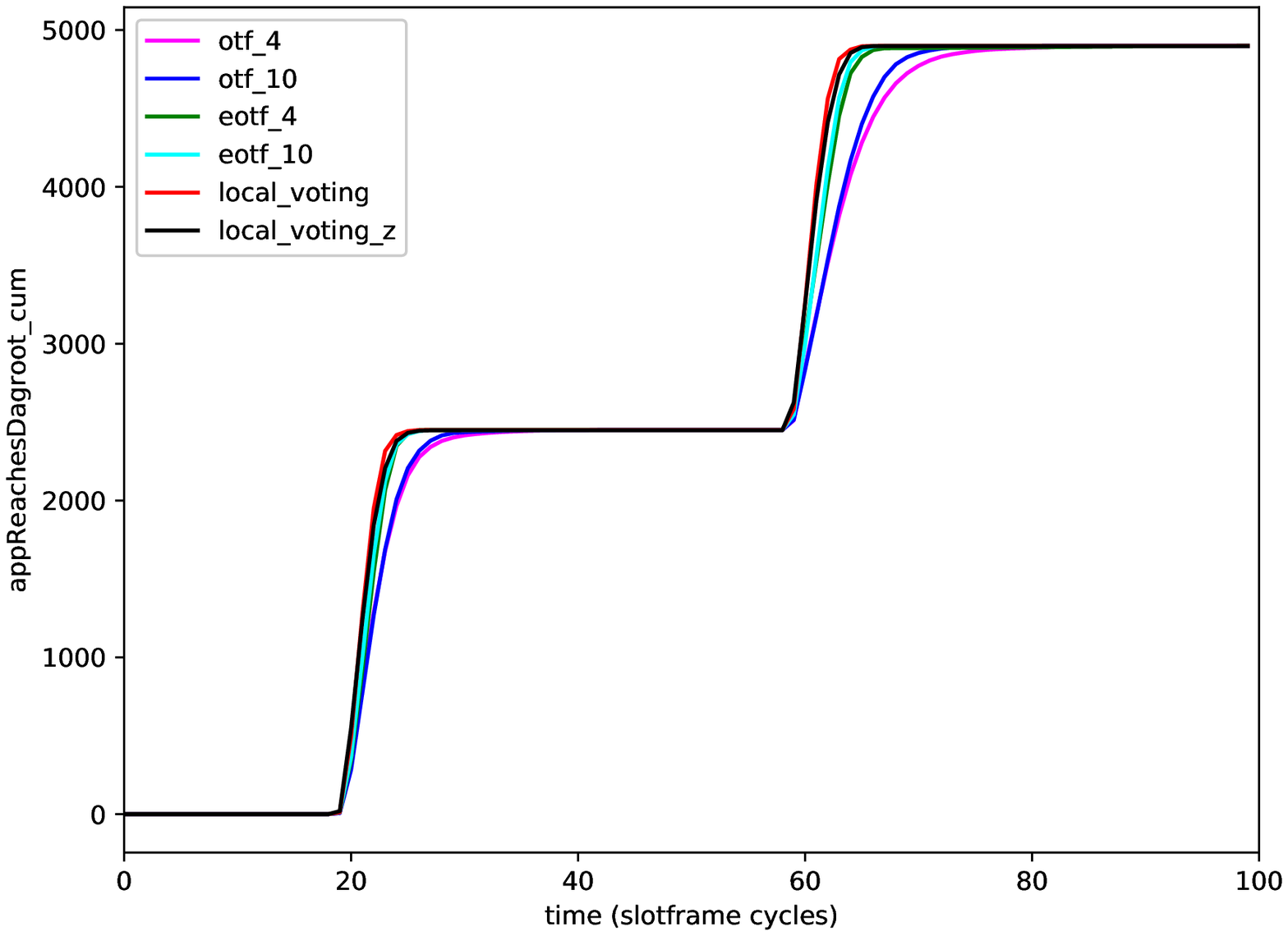}
  \caption{Number of packets that reach the root as a function of time, 50 packets per burst.}
  \label{appReachesDagroot_cum_vs_time_buf_100_par_3_pkt_50}
\end{figure}

\begin{figure}
  \centering
  \includegraphics[width=0.9\textwidth]{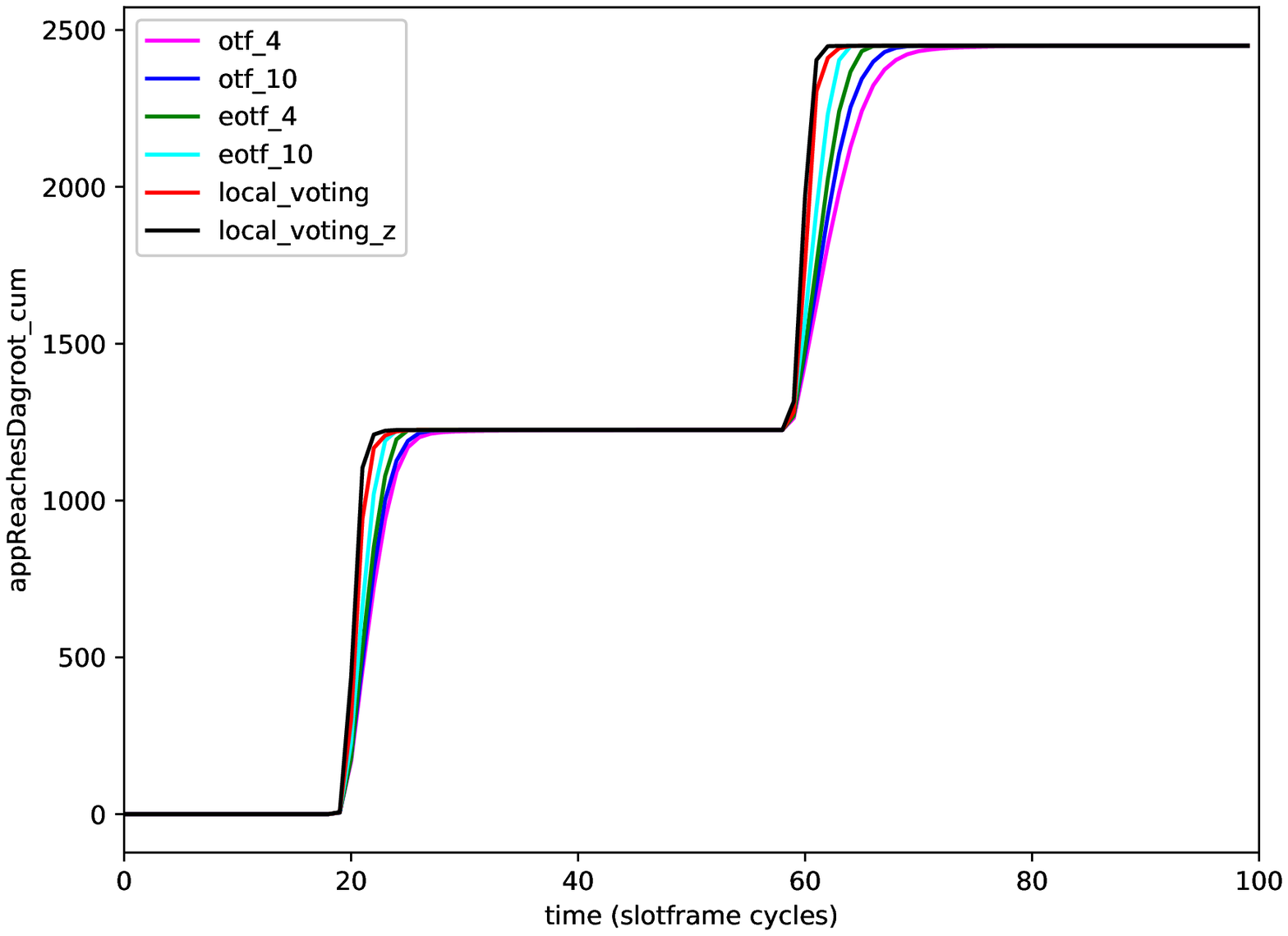}
  \caption{Number of packets that reach the root as a function of time, 25 packets per burst.}
  \label{appReachesDagroot_cum_vs_time_buf_100_par_3_pkt_25}
\end{figure}

\begin{figure}
  \centering
  \includegraphics[width=0.9\textwidth]{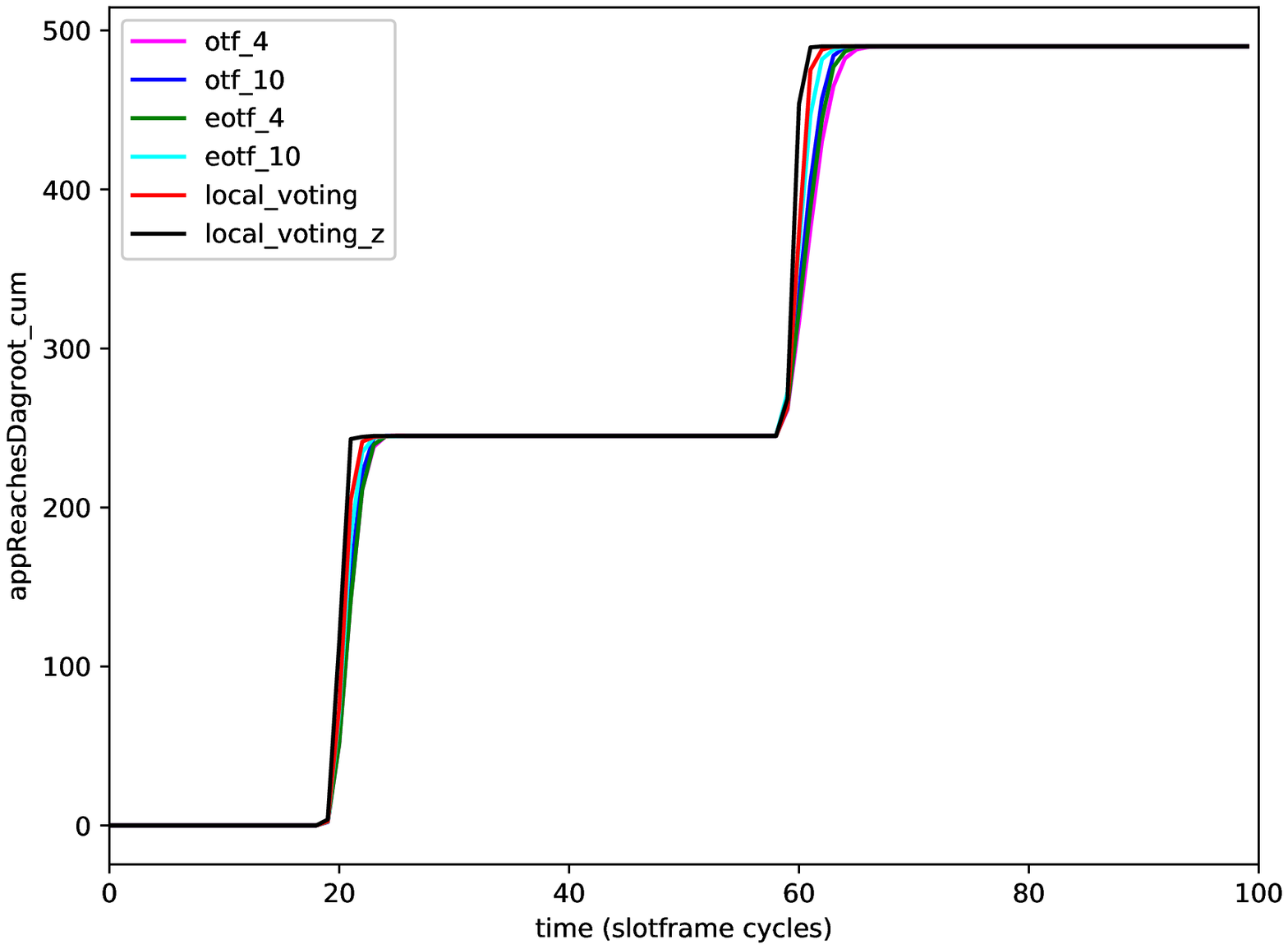}
  \caption{Number of packets that reach the root as a function of time, 5 packets per burst.}
  \label{appReachesDagroot_cum_vs_time_buf_100_par_3_pkt_5}
\end{figure}

Another important aspect of the performance of the algorithms is the energy
that is consumed for the delivering the data to the root.
Fig.~\ref{chargeConsumedPerRecv_vs_threshold_buf_100} depicts the energy that
is used per received packet, i.e., the fraction of the consumed energy over the
number of packets that were successfully delivered to the root.
As expected, as the number of packets per burst increases, the energy per
packet reduces.
We can also see that for small burst sizes OTF has the highest energy
consumption per packet, whereas for large burst sizes, E-OTF consumes the most
energy per packet.
Similar results are depicted in Fig.~\ref{chargeConsumed_vs_threshold_buf_100},
where the total energy consumption per simulation is depicted.
Here the energy consumed increases as the number of packets per burst increase,
which is expected, since there are more data transmissions.
We show again that for larger numbers of packets per burst, the E-OTF
algorithm consumes significantly more energy than the Local Voting and the OTF
algorithms.
The evolution of the energy consumption over time is given in
Figs.~\ref{chargeConsumed_vs_time_buf_100_par_3_pkt_80}--%
\ref{chargeConsumed_vs_time_buf_100_par_3_pkt_5}.
This confirms that E-OTF uses more energy than the other algorithms, whereas
Local Voting and OTF have similar consumption.
The conclusion is that Local Voting has performance in terms of delay similar
to E-OTF, but with an energy consumption similar to OTF, so it combines both
good delay performance and energy efficiency.

\begin{figure}
  \centering
  \includegraphics[width=0.9\textwidth]{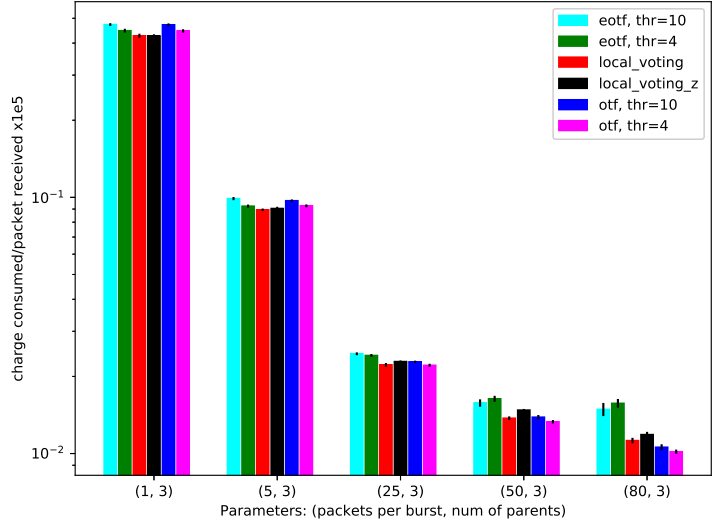}
  \caption{Energy consumption per received packet.}
  \label{chargeConsumedPerRecv_vs_threshold_buf_100}
\end{figure}

\begin{figure}
  \centering
  \includegraphics[width=0.9\textwidth]{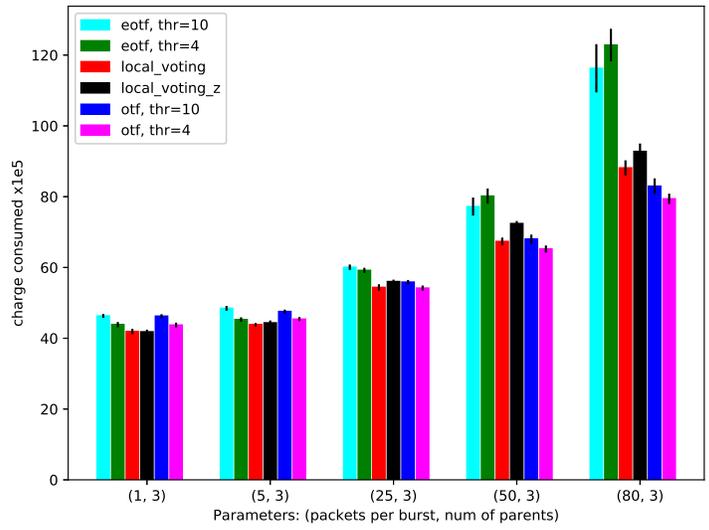}
  \caption{Energy consumption}
  \label{chargeConsumed_vs_threshold_buf_100}
\end{figure}

\begin{figure}
  \centering
  \includegraphics[width=0.9\textwidth]{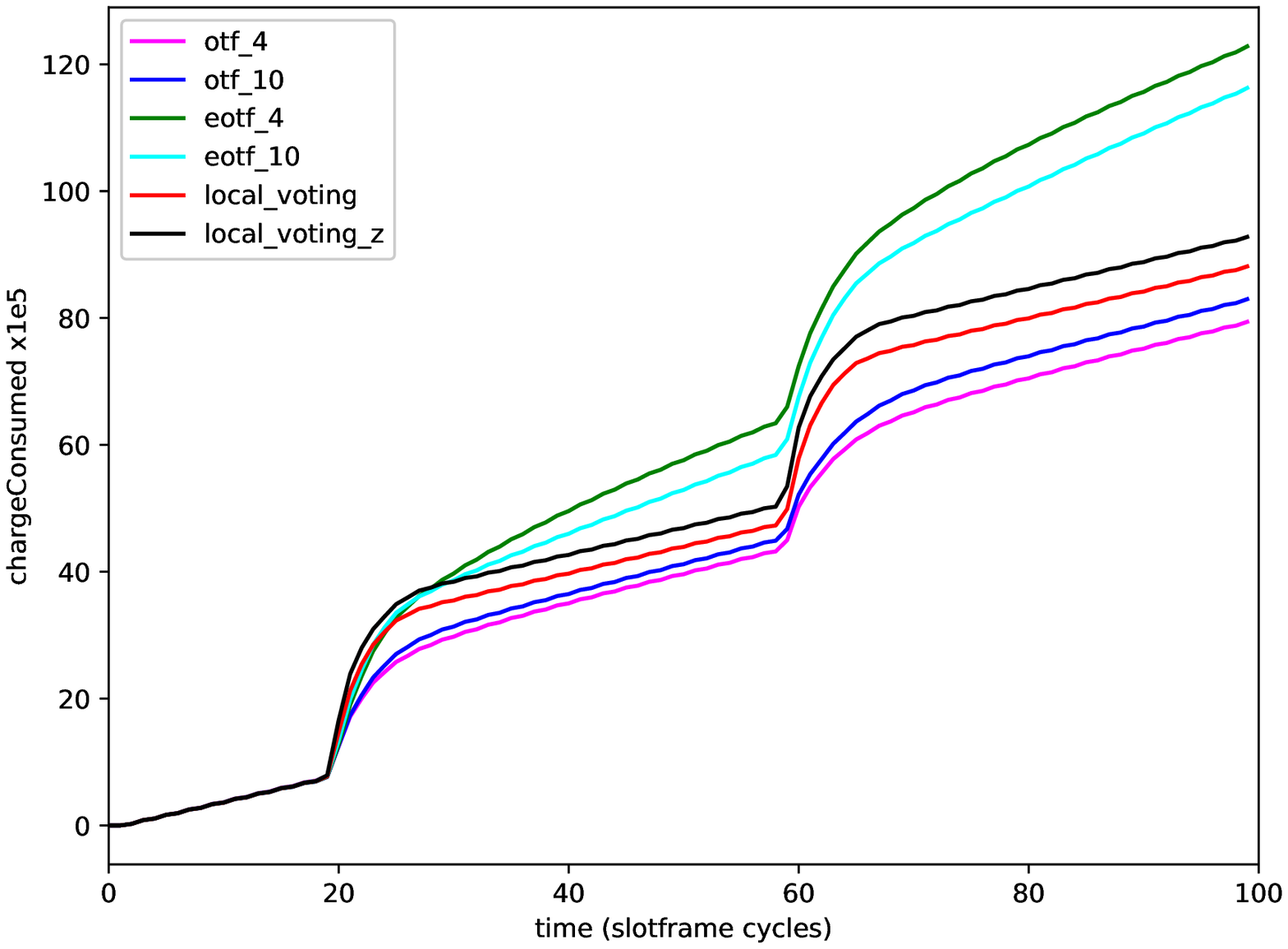}
  \caption{Evolution of the energy consumption over time, 80 packets per burst}
  \label{chargeConsumed_vs_time_buf_100_par_3_pkt_80}
\end{figure}

\begin{figure}
  \centering
  \includegraphics[width=0.9\textwidth]{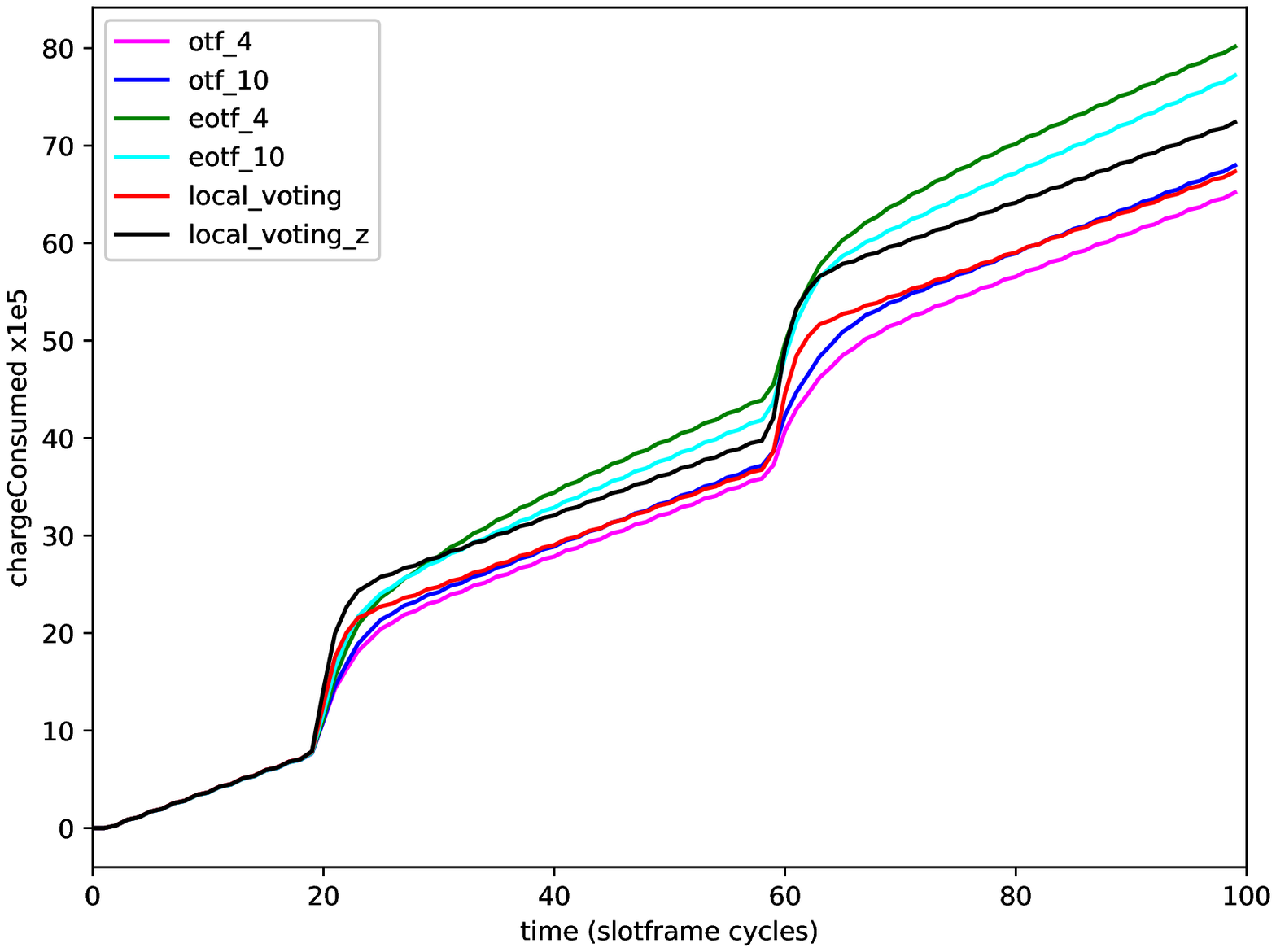}
  \caption{Evolution of the energy consumption over time, 50 packets per burst}
  \label{chargeConsumed_vs_time_buf_100_par_3_pkt_50}
\end{figure}

\begin{figure}
  \centering
  \includegraphics[width=0.9\textwidth]{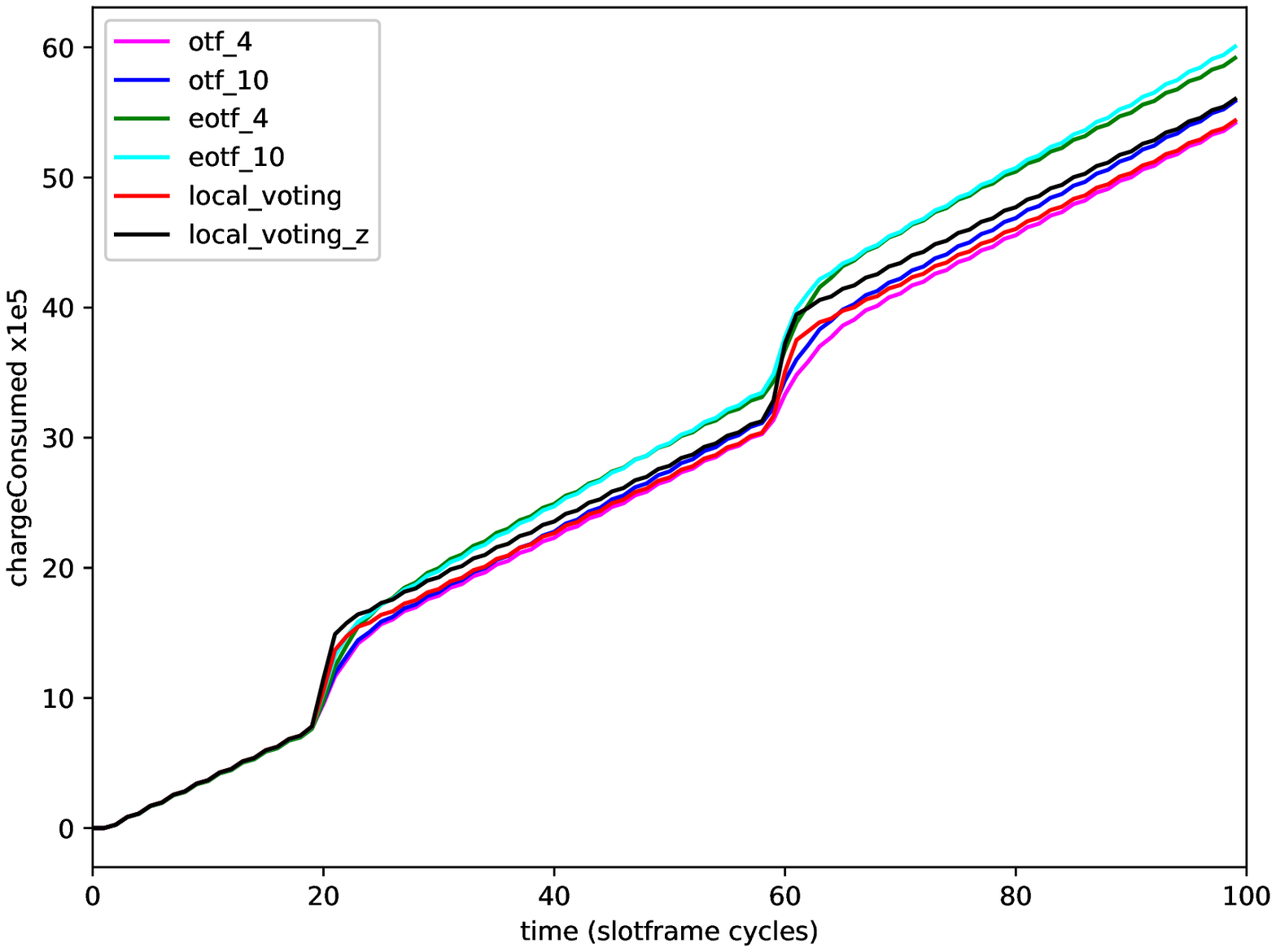}
  \caption{Evolution of the energy consumption over time, 25 packets per burst}
  \label{chargeConsumed_vs_time_buf_100_par_3_pkt_25}
\end{figure}

\begin{figure}
  \centering
  \includegraphics[width=0.9\textwidth]{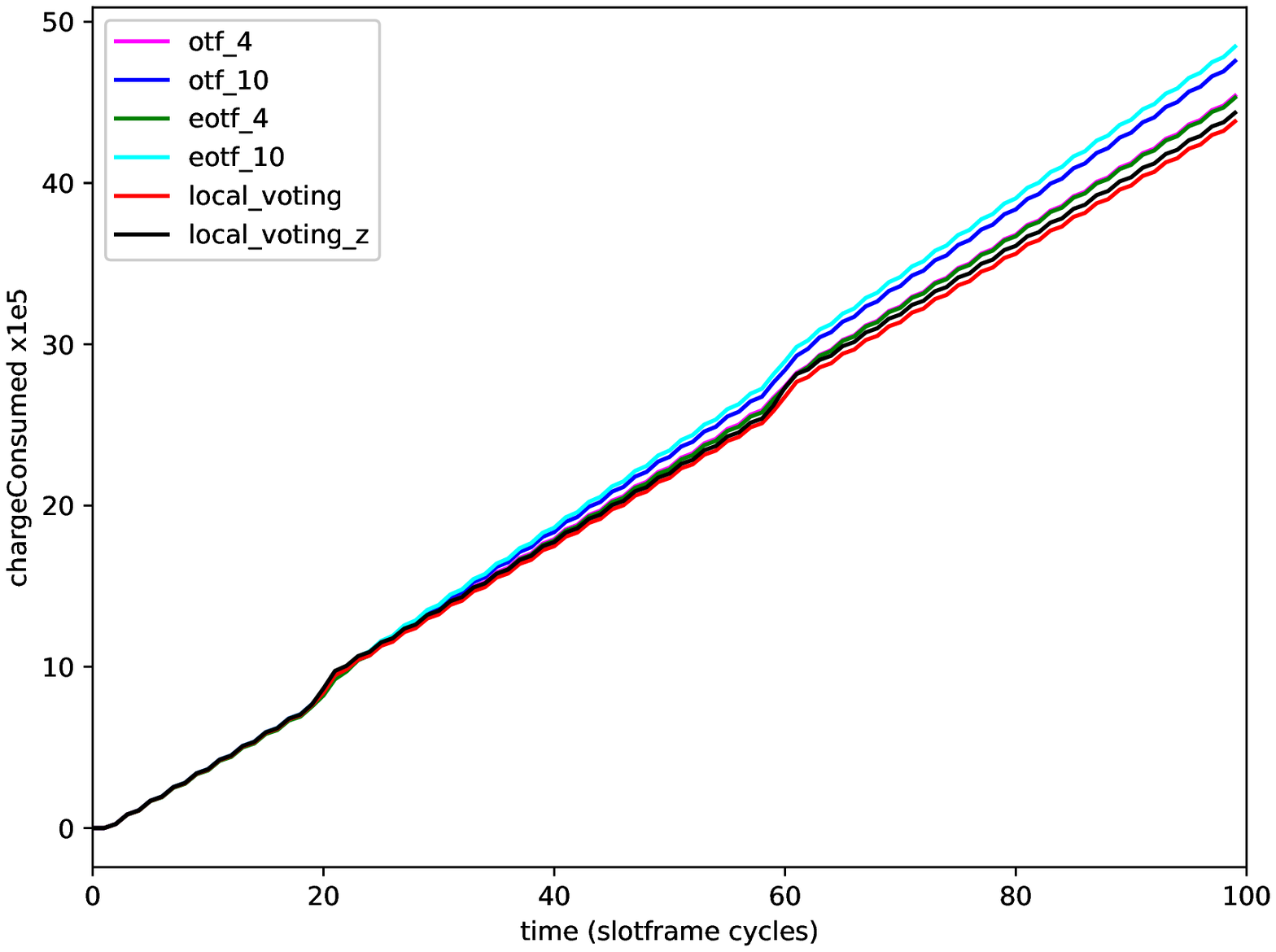}
  \caption{Evolution of the energy consumption over time, 5 packets per burst}
  \label{chargeConsumed_vs_time_buf_100_par_3_pkt_5}
\end{figure}

Fig.~\ref{txQueueFill_vs_threshold_buf_100} shows the average queue sizes among
all nodes in the network during the entire simulation, for each algorithm and
for each scenario.
A more detailed view is available in
Figs.~\ref{txQueueFill_vs_time_buf_100_par_3_pkt_80}
--\ref{txQueueFill_vs_time_buf_100_par_3_pkt_25}, where it is evident that the
increased efficiency of Local Voting makes the queue sizes to be reduced more
rapidly and all of the packets to reach their destinations faster.
In all cases the Local Voting algorithm achieves the smallest queue sizes,
which is the reason that it exhibits lower delay than the other algorithms.
This smaller queue size is also the reason behind the increased reliability of
the Local Voting algorithm compared to OTF and E-OTF.

\begin{figure}
  \centering
  \includegraphics[width=0.9\textwidth]{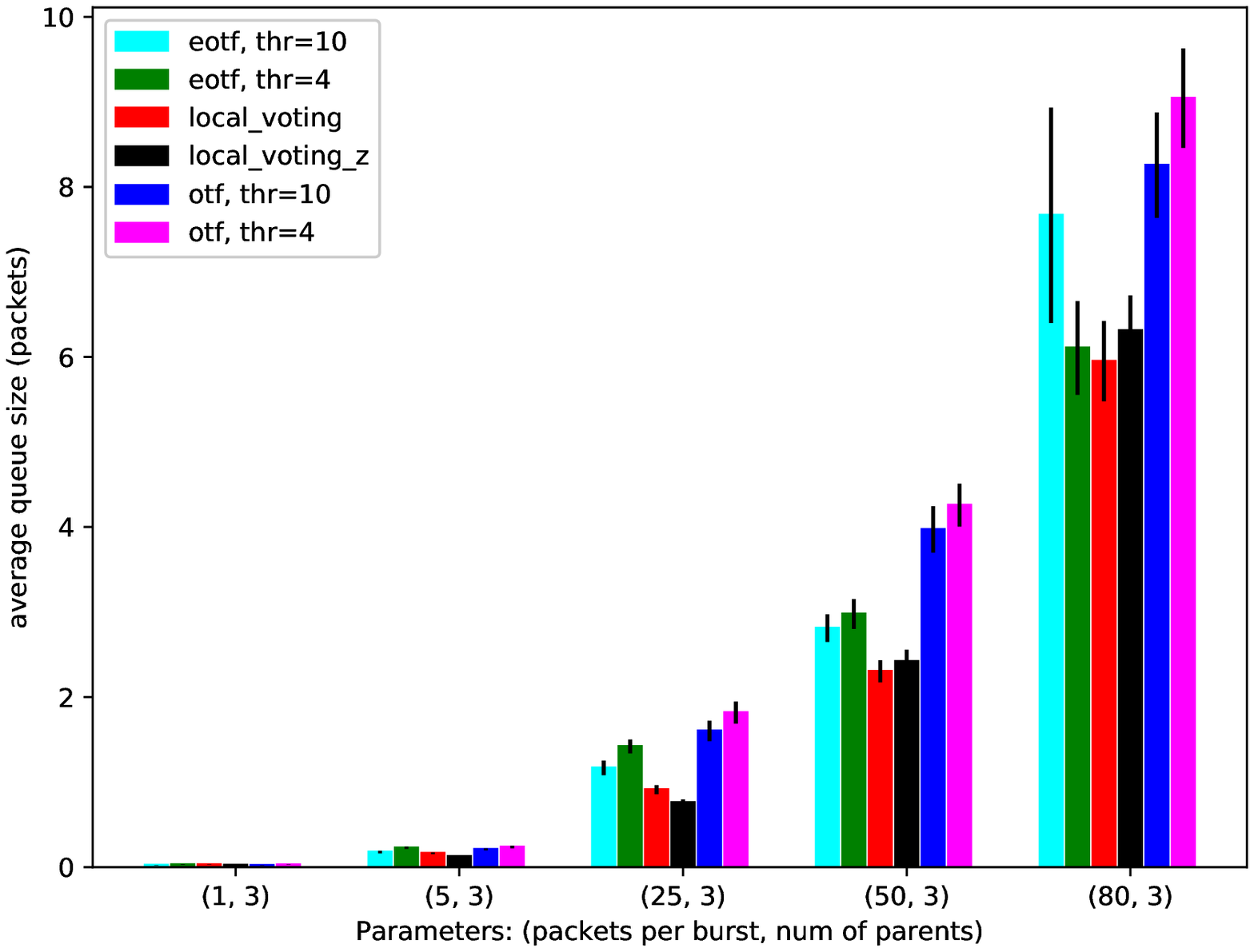}
  \caption{Average queue size (packets) for the entire simulation.}
  \label{txQueueFill_vs_threshold_buf_100}
\end{figure}

\begin{figure}
  \centering
  \includegraphics[width=0.9\textwidth]{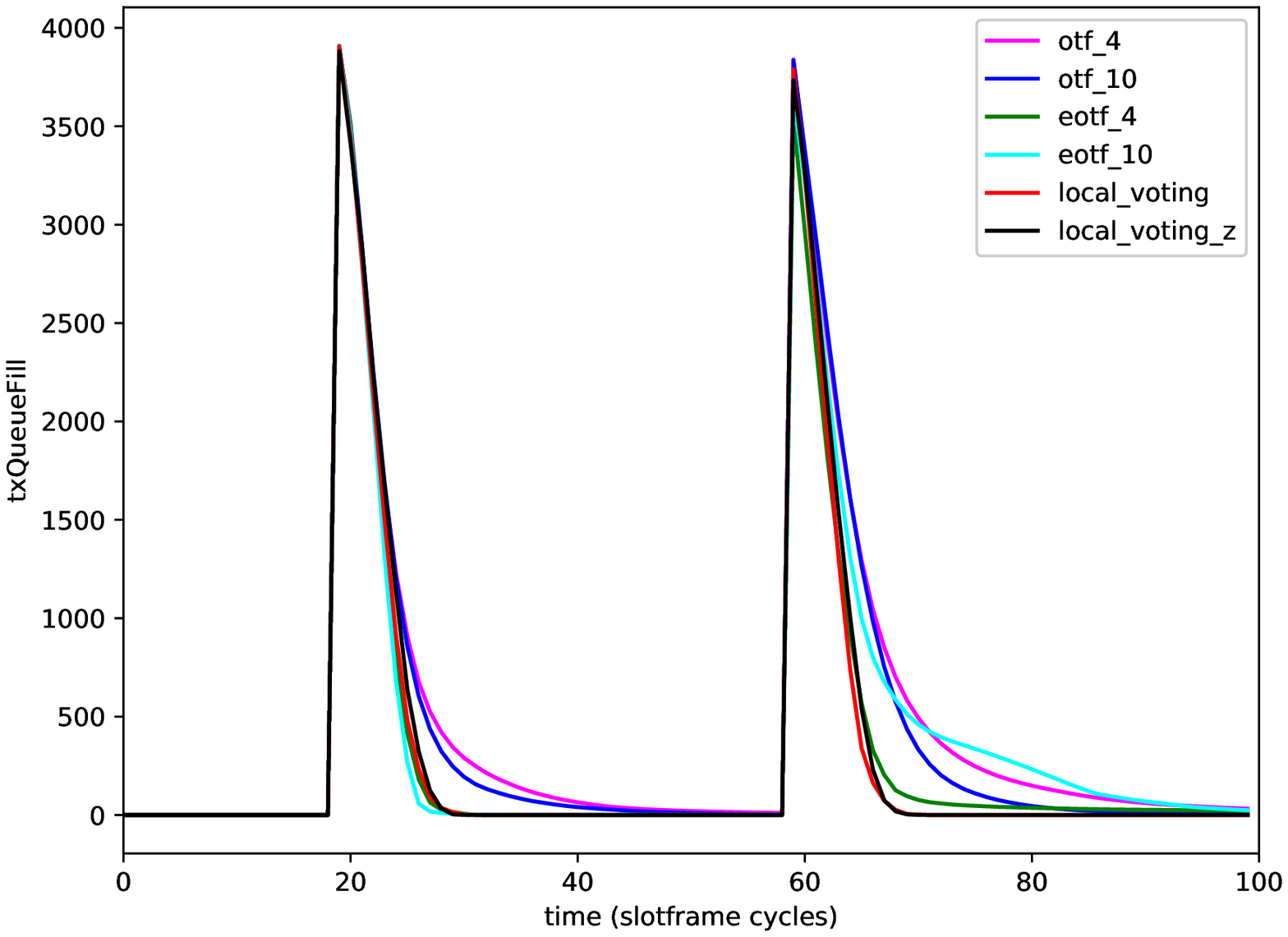}
  \caption{Evolution over time, 80 packets per burst.}
  \label{txQueueFill_vs_time_buf_100_par_3_pkt_80}
\end{figure}

\begin{figure}
  \centering
  \includegraphics[width=0.9\textwidth]{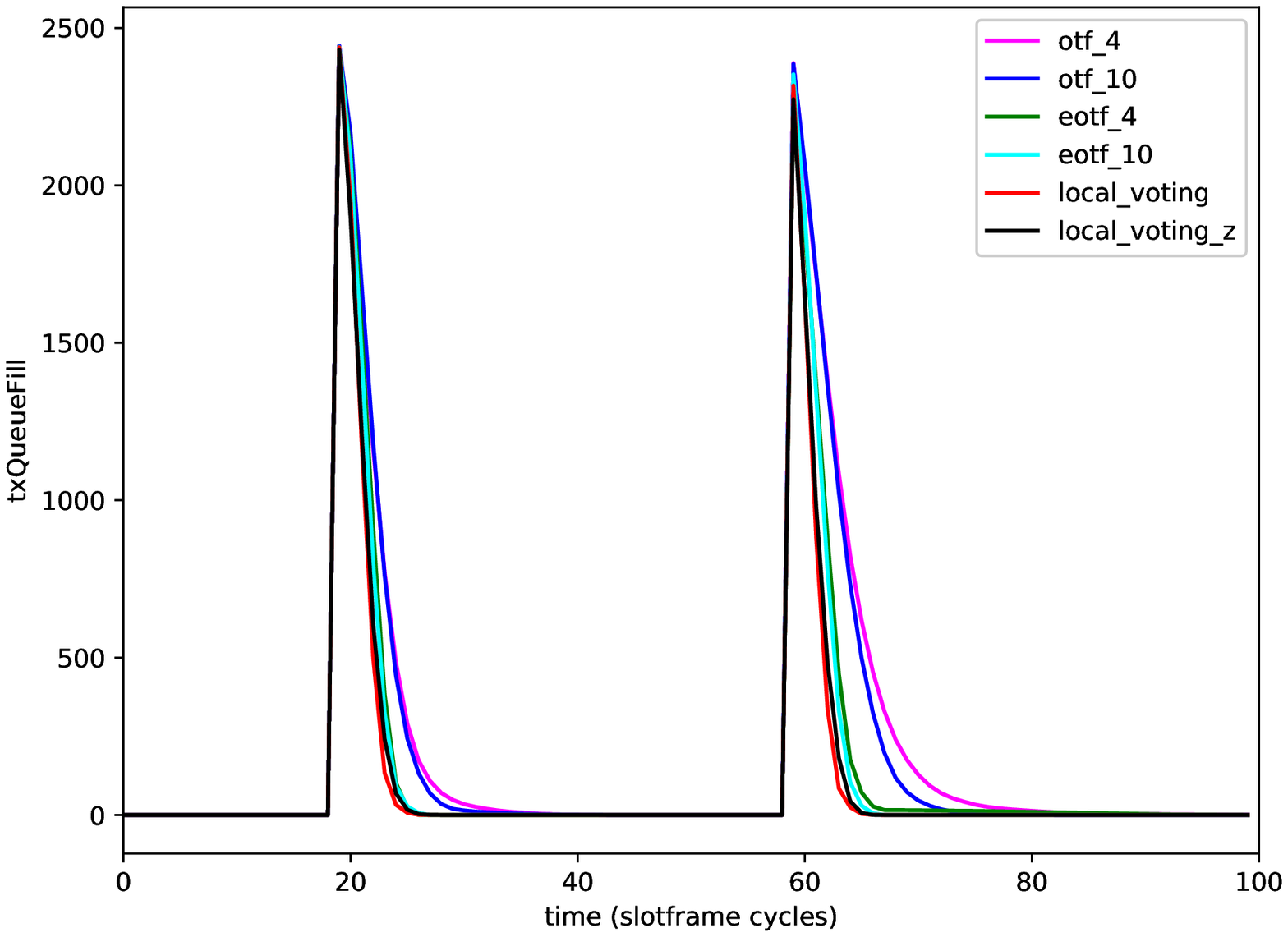}
  \caption{Evolution over time, 50 packets per burst.}
  \label{txQueueFill_vs_time_buf_100_par_3_pkt_50}
\end{figure}

\begin{figure}
  \centering
  \includegraphics[width=0.9\textwidth]{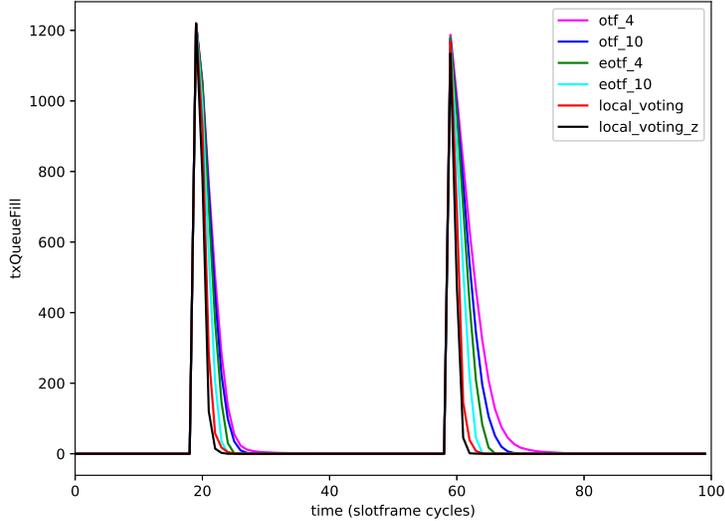}
  \caption{Evolution over time, 25 packets per burst.}
  \label{txQueueFill_vs_time_buf_100_par_3_pkt_25}
\end{figure}

In Fig.~\ref{LoadAllJain_vs_threshold_buf_100}--%
\ref{LoadAllG_vs_time_buf_100_par_3_pkt_80}, we can see the average fairness between
the nodes in the network, calculated on the load of each node (i.e.\ the ratio
of queue length over slot allocation), using two fairness
metrics, namely Jain's fairness index and the G fairness index.
The local voting algorithm has the best fairness in terms of load, which is
expected, since by design it tries to equalize the load throughout the
congested areas of the network.

\begin{figure}
  \centering
  \includegraphics[width=0.9\textwidth]{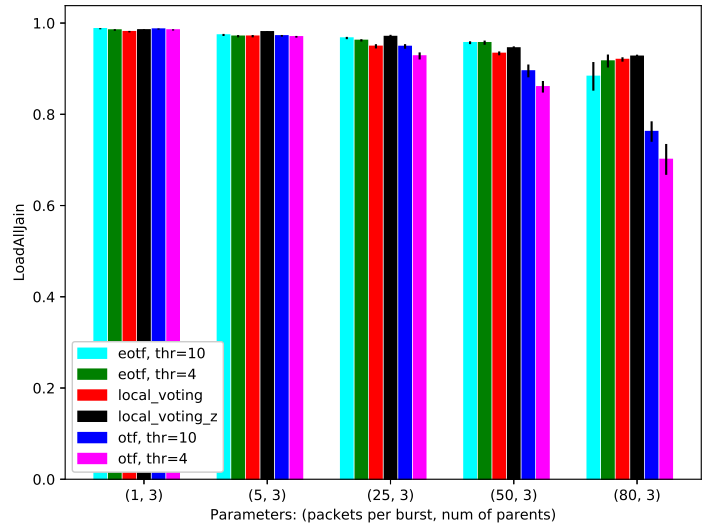}
  \caption{Fairness in load distribution, with Jain's fairness index, average.}
  \label{LoadAllJain_vs_threshold_buf_100}
\end{figure}

\begin{figure}
  \centering
  \includegraphics[width=0.9\textwidth]{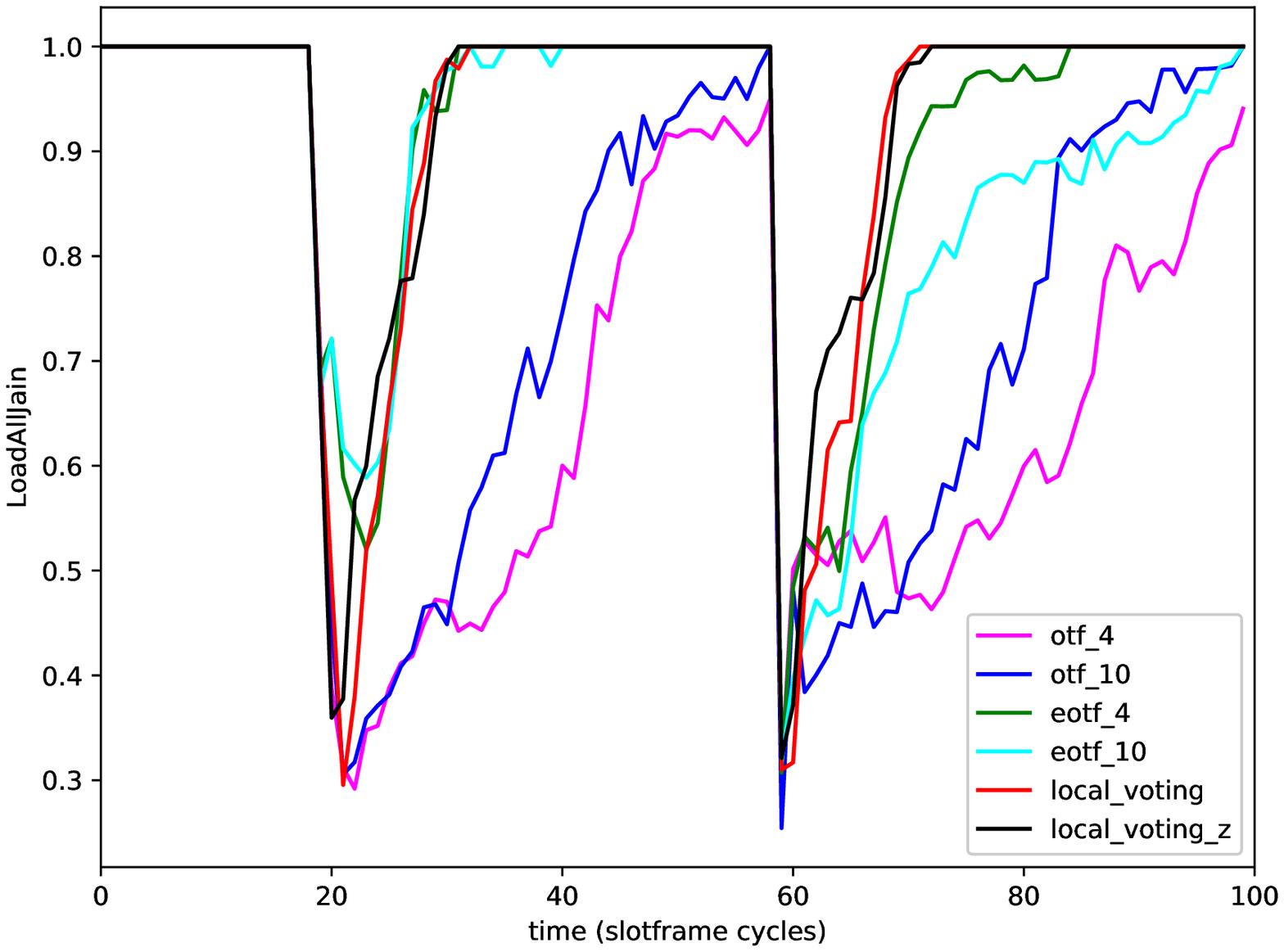}
  \caption{Fairness in load distribution, with Jain's fairness index, over time, 80 packets per burst.}
  \label{LoadAllJain_vs_time_buf_100_par_3_pkt_80}
\end{figure}

\begin{figure}
  \centering
  \includegraphics[width=0.9\textwidth]{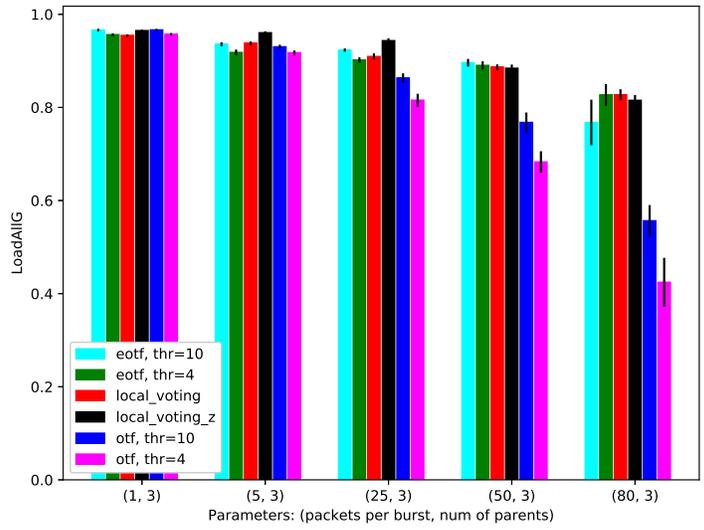}
  \caption{Fairness in load distribution, with G fairness index, average.}
  \label{LoadAllG_vs_threshold_buf_100}
\end{figure}

\begin{figure}
  \centering
  \includegraphics[width=0.9\textwidth]{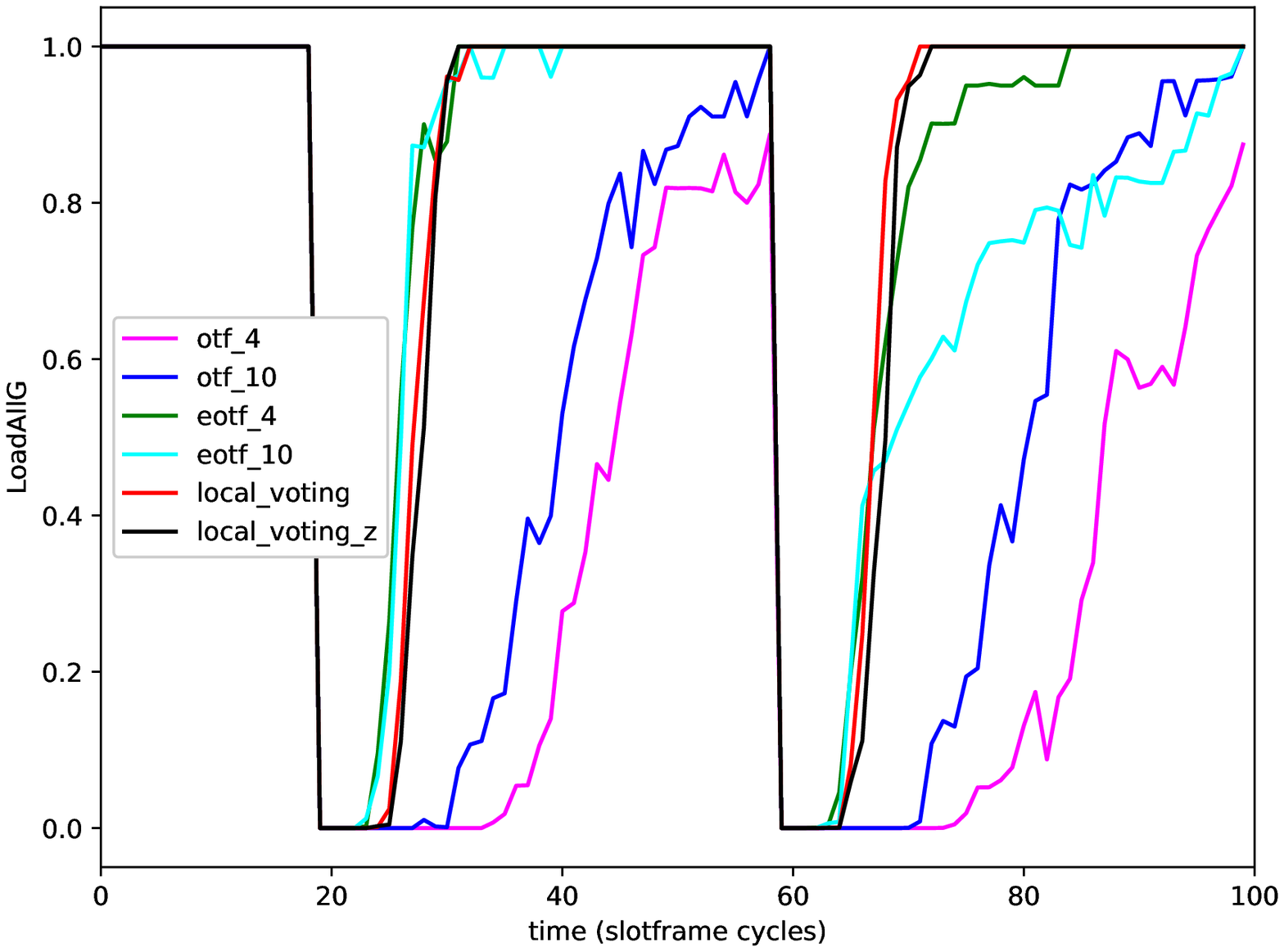}
  \caption{Fairness in load distribution, with G fairness index, over time, 80 packets per burst.}
  \label{LoadAllG_vs_time_buf_100_par_3_pkt_80}
\end{figure}

\subsection{Uniform traffic experiments}

This subsections contains the results of the uniform traffic experiment, where the nodes
transmit at a constant rate, with some variability in the traffic generation time to
avoid synchronization issues.

In Fig.~\ref{steady_max_latency_vs_threshold_buf_100} we can see the maximum latency
for each scenario.
We can see that in this scenario the advantage of local voting z over
the previous version of local voting in terms of maximal delay.
Specifically, local voting z
has the smallest maximum delay compared to all the other algorithms. This can be explained,
since for local voting a large queue is necessary for increasing the slot allocation, whereas
in the case of local voting z, the slot allocation also tracks the new packets that are
generated at each round, so that the buffer-bloat problem can be avoided.
In a sense local
voting z considers the ongoing rate of traffic that must be delivered, in addition to the
current buffer size, whereas local voting (without z) only considers the buffer.
The difference is even more apparent in Fig.~\ref{steady_latency_vs_threshold_buf_100},
where the average latency is depicted.

\begin{figure}
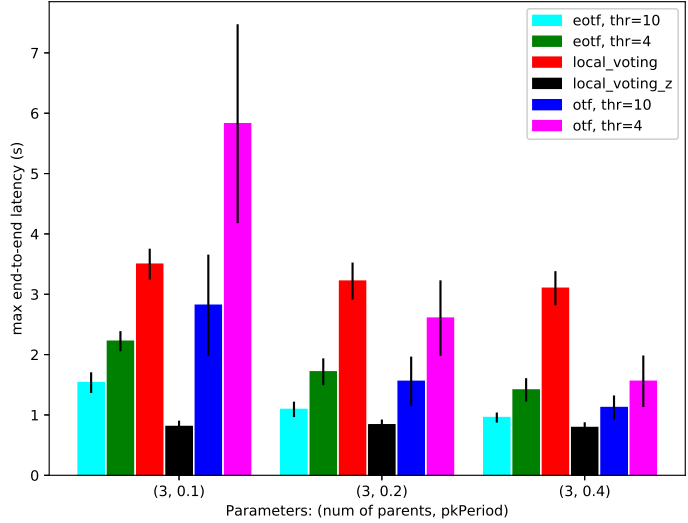

  \centering
  \includegraphics[width=0.9\textwidth]{{{steady_max_latency_vs_threshold_buf_100}}}
  \caption{The maximum average latency for the different scenarios.}
  \label{steady_max_latency_vs_threshold_buf_100}
\end{figure}

\begin{figure}
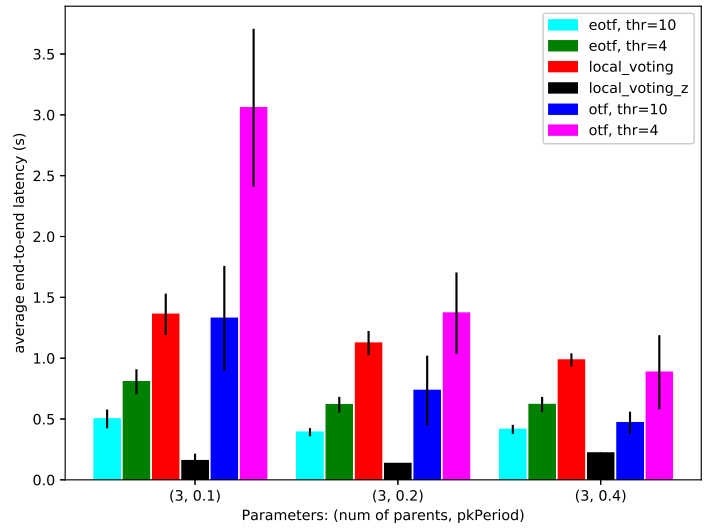

  \centering
  \includegraphics[width=0.9\textwidth]{{{steady_latency_vs_threshold_buf_100}}}
  \caption{The average latency for the different scenarios.}
  \label{steady_latency_vs_threshold_buf_100}
\end{figure}

The evolution over the latency over time may be seen in Fig.\ref{steady_latency_vs_time_buf_100_par_3_pkt_1_per_0.1}.
Similar results are available for the other scenarios as well.

\begin{figure}
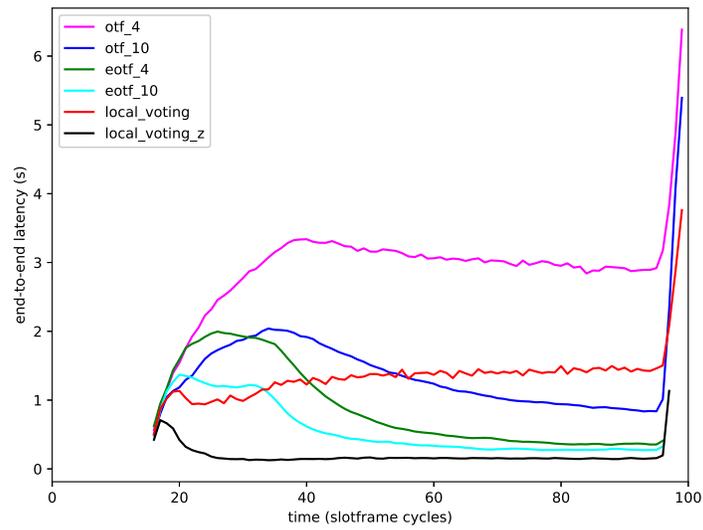

  \centering
  \includegraphics[width=0.9\textwidth]{{{steady_latency_vs_time_buf_100_par_3_pkt_1_per_0.1}}}
  \caption{The average latency over time for an packet inter-arrival time of 0.1 seconds.}
  \label{steady_latency_vs_time_buf_100_par_3_pkt_1_per_0.1}
\end{figure}

However, the improved performance in terms of latency comes at a cost of larger energy
consumption (Fig.~\ref{steady_chargeConsumedPerRecv_vs_threshold_buf_100},~\ref{steady_chargeConsumed_vs_threshold_buf_100}).

\begin{figure}
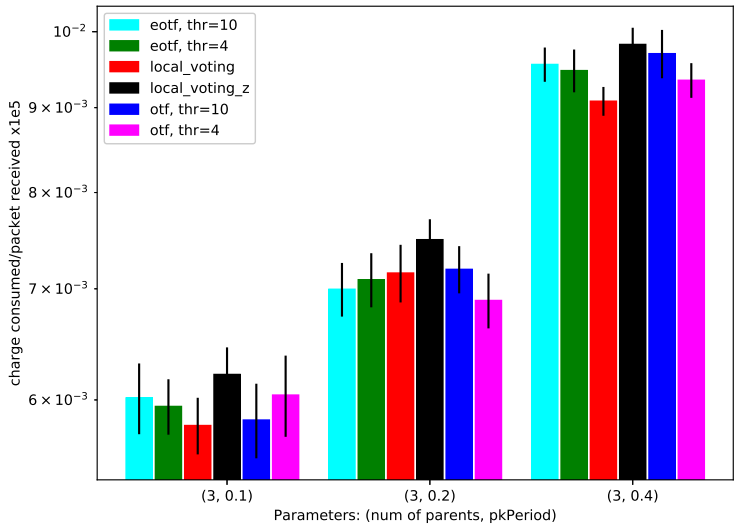

  \centering
  \includegraphics[width=0.9\textwidth]{{{steady_chargeConsumedPerRecv_vs_threshold_buf_100}}}
  \caption{The charge consumed per received packet for the different scenarios.}
  \label{steady_chargeConsumedPerRecv_vs_threshold_buf_100}
\end{figure}

\begin{figure}
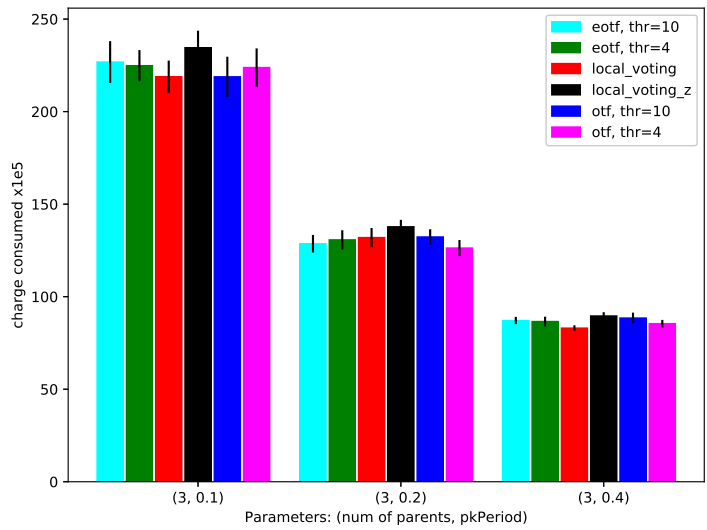

  \centering
  \includegraphics[width=0.9\textwidth]{{{steady_chargeConsumed_vs_threshold_buf_100}}}
  \caption{The total charge consumed for the different scenarios.}
  \label{steady_chargeConsumed_vs_threshold_buf_100}
\end{figure}

The reliability of all algorithms except OTF was perfect in all cases (Fig.~\ref{steady_reliability_vs_threshold_buf_100}.

\begin{figure}
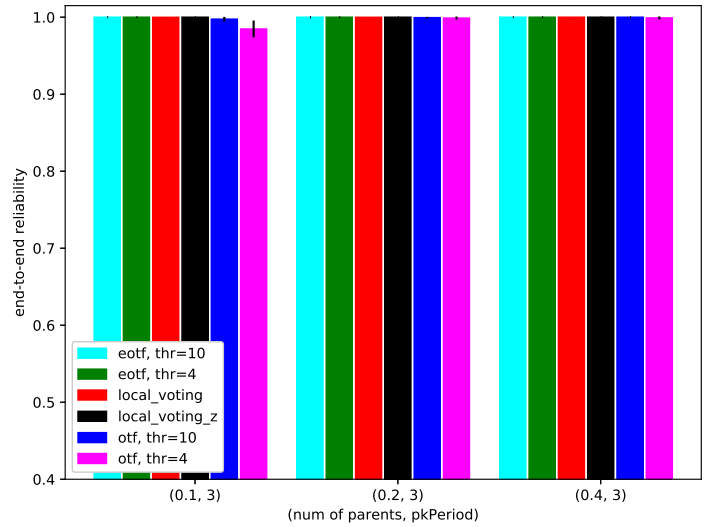

  \centering
  \includegraphics[width=0.9\textwidth]{{{steady_reliability_vs_threshold_buf_100}}}
  \caption{The reliability (ratio of generated packets that reach their destination) for the different scenarios.}
  \label{steady_reliability_vs_threshold_buf_100}
\end{figure}

The lower delay of the local voting z algorithm can be easily explained by seeing Fig.~\ref{steady_max_txQueueFill_vs_threshold_buf_100},
Fig.~\ref{steady_txQueueFill_vs_threshold_buf_100}, and Fig.~\ref{steady_txQueueFill_vs_time_buf_100_par_3_pkt_1_per_0.1},
where it is apparent that the queue sizes
are much smaller for local voting z, resulting in the reduction of the latency.

\begin{figure}
  \centering
  \includegraphics[width=0.9\textwidth]{{{steady_max_txQueueFill_vs_threshold_buf_100}}}
  \caption{The maximum average queue size for the different scenarios.}
  \label{steady_max_txQueueFill_vs_threshold_buf_100}
\end{figure}

\begin{figure}
  \centering
  \includegraphics[width=0.9\textwidth]{{{steady_txQueueFill_vs_threshold_buf_100}}}
  \caption{The average queue size for the different scenarios.}
  \label{steady_txQueueFill_vs_threshold_buf_100}
\end{figure}

\begin{figure}
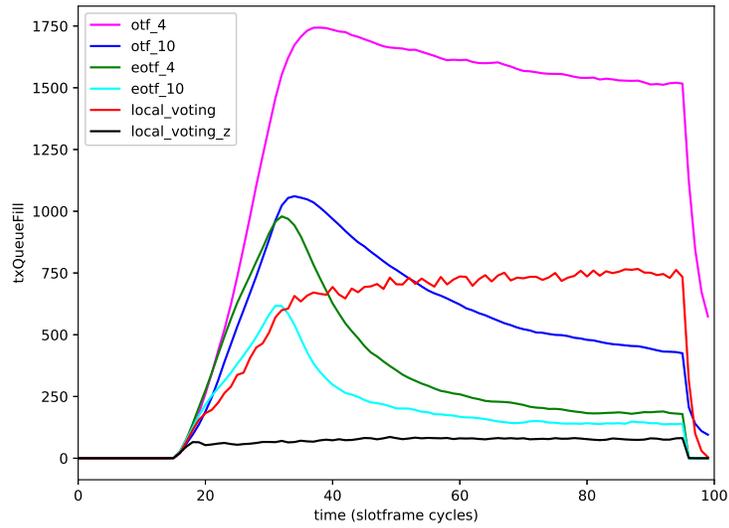

  \centering
  \includegraphics[width=0.9\textwidth]{{{steady_txQueueFill_vs_time_buf_100_par_3_pkt_1_per_0.1}}}
  \caption{The average queue size over time for an packet inter-arrival time of 0.1 seconds.}
  \label{steady_txQueueFill_vs_time_buf_100_par_3_pkt_1_per_0.1}
\end{figure}

Additional results are available at github repository\footnote{%
  As an online addition to this article, the source code is
  available at \url{https://github.com/djvergad/local_voting_tsch}
}, that are omitted in this
paper due to space limitations.

\section{Conclusions}
\label{conc}
We proposed a new distributed bandwidth reservation algorithm called
\emph{Local Voting} which balances the load between links in 6TiSCH networks.
The algorithm calculates the number of cells to be added or released by 6top
while considering the collision-free constraints.
In this way, it adapts the schedule to the network conditions in 6TiSCH
networks, equalizes the load in congested areas, that as expected provides
efficient resource allocation.
We showed that optimal schedules are maximal and balanced, and these are the
two design goals of LV\@.
Extensive simulation results show that LV combines the good delay performance
of E-OTF and the energy efficiency of OTF, while outperforming them in terms of
reliability and fairness.
To summarize, we proved the advantage of load balancing when performing link
scheduling in 6TiSCH networks, proposed Local Voting for distributed bandwidth
reservation in 6TiSCH networks, and we demonstrated by simulations that the
Local Voting algorithm shows an overall very good performance in comparison
with other state-of-the-art algorithms.
In addition, the new variant of the algorithm, named \emph{local voting z},
achieves lower latency that the other algorithms we compared with, even
in the steady-state scenario.

\section*{References}

\bibliography{refer}

\begin{thebibliography}{10}
\bibitem{cisco}
D.~Evans, The internet of things how the next evolution of the internet is
  changing everything, Cisco 1.
\end{thebibliography}

\end{document}